%% file: neurips_2026.tex
\documentclass{article}
\usepackage[utf8]{inputenc}

\PassOptionsToPackage{numbers, compress}{natbib}
 \usepackage[preprint]{neurips_2026} 

\usepackage[utf8]{inputenc} 
\usepackage[T1]{fontenc}    
\usepackage{hyperref}       
\usepackage{url}            
\usepackage{booktabs}       
\usepackage{amsfonts}       
\usepackage{nicefrac}       
\usepackage{microtype}      
\usepackage{xcolor}         

\usepackage{tcolorbox}
\tcbuselibrary{skins,breakable}
\usepackage{graphicx}
\usepackage{booktabs}
\usepackage{siunitx}
\usepackage{xcolor}
\usepackage{multirow}
\usepackage{makecell}
\usepackage{ragged2e}
\usepackage{amsmath}
\usepackage{array}
\usepackage{multirow}
\usepackage{xurl}
\usepackage{fancyvrb}
\usepackage{xcolor}

\usepackage{colortbl}
\definecolor{lightgreen}{HTML}{e0ffcd}
\definecolor{lightyellow}{HTML}{ffebbb}
\definecolor{lightred}{HTML}{ffcccc}
\definecolor{lightblue}{RGB}{173,216,230}   
\definecolor{lightpink}{RGB}{248, 200, 220}
\definecolor{coral}{RGB}{255,127,80} 
\definecolor{pink}{RGB}{255,192,203}
\usepackage{wrapfig}
\usepackage{caption}
\usepackage{algorithm}
\usepackage{algorithmicx}
\usepackage{algpseudocode}
\usepackage{amsthm}
\newtheorem{proposition}{Proposition}
\usepackage{float}

\definecolor{lightgray}{gray}{0.95}
\definecolor{darkblue}{rgb}{0,0,0.7}

\title{ReviewGuard: Aligning LLM-Assisted Peer Review with Long-Term Scientific Impact}

\author{%
  Abdur Rasool, Xiaohui Huang, Yanqing Hu, Linyi Yang \\
  Department of Statistics and Data Science \\
  Southern University of Science and Technology \\
  Shenzhen 518055, China \\
  \texttt{\{abdurrasool, huangxh, huyq, yangly6\}@sustech.edu.cn}
}

\begin{document}

\maketitle

\begin{abstract} 
Peer review is central to scientific quality control, yet it can undervalue papers that later achieve substantial citation impact. While frontier large language models have shown promise in automating aspects of peer review, they primarily mimic human reviewer preferences rather than predict long-term scientific value. We introduce \textbf{ReviewGuard}, a two-stage framework that aligns LLM-generated reviews with citation-based estimates of long-term scientific impact rather than contemporaneous reviewer judgments. On 20,861 AI/ML papers from OpenReview augmented with Semantic Scholar citation data, ReviewGuard achieves a Spearman correlation of $\rho=0.776$ with future citations on rejected-then-published papers, outperforming human reviewers ($\rho=0.492$) and a supervised Expert model ($\rho=0.681$). Under the same decision threshold, ReviewGuard flags $10.2\%$ of high-impact rejected papers, compared with $1.8\%$ for human reviewers, corresponding to a $5.6\times$ improvement.  Our results demonstrate that impact-aligned reinforcement learning can provide editors with a complementary signal for identifying high-potential work, without replacing human judgment. 


\end{abstract}

\section{Introduction}
\label{sec:introduction}

Peer review is the cornerstone of scientific quality control, yet it is increasingly strained by exponential submission growth and well-documented inconsistencies in evaluating long-term impact~\cite{bornmann2011peer, smith2006peer}. High-impact papers are frequently rejected due to conservative scoring, reviewer bias, or limited bandwidth, only to later accumulate substantial citations~\cite{siler2015measuring, ponomarev2014predicting, calder2013authoritarian, powell2020does}. This ``rejection-resilience'' phenomenon \citep{cortes2021inconsistency, wang2013quantifying, sinatra2016quantifying} is particularly pronounced in machine learning conferences, where rejection rates at top venues often exceed 75\%~\cite{tran2020analyzing}. The consequences are substantial: delayed dissemination of promising ideas and slower scientific progress. Traditional interventions—structured review forms, reviewer training, and open review platforms—have yielded only marginal improvements, as they remain fundamentally constrained by inherent human cognitive limitations and persistent psychological biases in accurately predicting long-term impact from a single reading~\cite{bruce2016peerreviewinterventions, schroter2004training, lee2013bias, xie2024reviewerscores}.

Parallel advances in large language models (LLMs) have opened new frontiers for automating peer review~\cite{vaswani2017attention, openai2023gpt4, hosseini2023fighting, liu2023reviewergpt}. Recent work has explored instruction-tuned LLMs as reviewers~\cite{zhou2024reliable}, specialized review-generation models like DeepReview~\cite{zhu2025deepreview}, multi-agent systems~\cite{jiang2024agentreview}, and prompt engineering strategies~\cite{darcy2024marg}. However, these systems remain optimized to \textbf{\emph{reproduce current reviewer preferences} rather than predict \emph{future scientific impact}}. Consequently, they inherit the same biases and limitations of the human process they aim to assist.

A more recent direction is reinforcement learning from human feedback (RLHF)~\cite{ouyang2022training, rafailov2024direct}, which has proven highly effective for aligning LLMs with human preferences. Yet its application to peer review remains virtually unexplored. Group Relative Policy Optimization (GRPO)~\cite{shao2024deepseekmath}, a variant that eliminates the need for a separate critic model, has never been applied to scientific impact prediction \citep{shao2024deepseekmath, deepseek2025deepseekmathv2}. These \textbf{limitations motivate a fundamental shift}: instead of asking whether an LLM can mimic a human reviewer, we ask whether it can predict a paper's future scientific impact. Human reviews are imperfect proxies for long-term value; citation counts, while noisy \citep{price1976citation, clauset2017scientist}, offer a direct, quantifiable measure of real-world influence~\cite{sinatra2016quantifying, clauset2017scientist}.

To this end, we introduce \textbf{ReviewGuard}, a two-stage framework that aligns an LLM with long-term citation impact. Stage 1 performs Low-Rank Adaptation (LoRA)-based supervised fine-tuning of Qwen2-7B-Instruct \citep{hu2022lora, dettmers2024qlora, qwen2} on over 20,861 (20k) high-quality peer reviews to create an Expert model. Stage 2 applies GRPO with a novel citation-based reward that normalizes future citations to the rating scale, penalizes deviation from normalized impact, and provides a substantial bonus for confidently identifying exceptionally high-impact papers (>150 citations). Unlike prior work, ReviewGuard explicitly optimizes for long-term scientific value rather than mimicking current reviewer preferences. We experimented ReviewGuard on 20,861 AI/ML papers from OpenReview augmented with Semantic Scholar citation data.  On the top 1,000 most-cited rejected-then-published papers, ReviewGuard achieves a Spearman correlation of $\rho=0.776$ with future citations, substantially outperforming human reviewers ($\rho=0.492$) and our Expert model ($\rho=0.681$). It rescues $10.2\%$ of high-impact rejected papers—a $5.6\times$ improvement over human reviewers ($1.8\%$)—and delivers more calibrated ratings across impact deciles. Our main contributions are threefold:

\begin{itemize}
    
    \item \textbf{ReviewGuard}, a two-stage LLM-based peer-review framework that combines supervised fine-tuning with citation-based GRPO to align review ratings with long-term citation impact.

    \item \textbf{A rejected-paper impact prediction benchmark} built from peer-review data enriched with future citation counts, linking initially rejected submissions to their later published versions.

    \item \textbf{Empirical validation} showing that ReviewGuard achieves stronger correlation with future citations than both human reviewers and the supervised Expert model, while more effectively flagging high-impact papers that were initially rejected.
\end{itemize}

Together, these contributions establish ReviewGuard as an impact-aligned framework for LLM-assisted peer review: it introduces a citation-based reinforcement learning objective, provides a benchmark for evaluating rejected-paper impact prediction, and empirically improves alignment with future citations over both human reviewer scores and the supervised Expert model. 

\raggedbottom
\section{Preliminaries}
\label{sec:Pre}

\paragraph{Problem Formulation.} Let $\mathcal{D} = \{(p_i, r_i^h, c_i)\}_{i=1}^{N}$ be a dataset of peer-reviewed papers, where $p_i$ denotes the paper content (title, abstract, and main text) represented as a token sequence $p_i = (x_1, \ldots, x_{T_i})$, $r_i^h \in [1,10]$ denotes the average human reviewer rating, and $c_i \in \mathbb{R}_{\ge 0}$ denotes the future citation count collected from Semantic Scholar after publication. Our core technical optimization goal is to learn a policy $\pi_\theta: \mathcal{P} \rightarrow [1,10]$ that maps a paper to a predicted rating $r = \pi_\theta(p)$ such that the rating is strongly correlated with the paper's future scientific impact. Formally, we aim to maximize the expected Spearman rank correlation \citep{spearman1904proof}:

\begin{equation}
\max_{\theta} \; \mathbb{E}_{p \sim \mathcal{D}} \left[ \rho\!\left( \pi_\theta(p), c(p) \right) \right],
\label{eq:objective}
\end{equation}

where $\rho(\cdot, \cdot)$ denotes the Spearman rank correlation coefficient. Since direct optimization of $\rho$ is non-differentiable and computationally intractable, we adopt a practical reinforcement learning framework where the policy generates a rating $r \sim \pi_\theta(\cdot \mid p)$ and we optimize a surrogate reward function $R(r, c)$. While $R$ does not guarantee direct maximization of $\rho$, we empirically show it leads to increased Spearman correlation (Section~\ref{sec:main_results}) \citep{sutton2018reinforcement, mnih2015human}.

\paragraph{Base Model and Supervised Fine-Tuning.} We adopt Qwen2-7B-Instruct~\cite{qwen2} as our base model. To create a high-quality initial peer-reviewer, we perform Supervised Fine-Tuning (SFT) on a dataset of 20,861 papers and peer reviews collected from OpenReview (primarily ICLR, NeurIPS, and top-tier AI venues), producing the \emph{Expert model}. 


\section{Method}
\label{sec:method}

We propose \textbf{ReviewGuard}, a novel two-stage framework for aligning peer-review ratings with long-term scientific impact. Motivated by extensive evidence that peer review can undervalue papers that later receive substantial citations~\cite{siler2015measuring, calder2013authoritarian}, ReviewGuard uses citation counts as an ex post impact signal. As shown in Figure~\ref{fig:ReviewGuard}, \textbf{Stage 1} constructs a dataset of rejected-then-published papers paired with citation counts and performs LoRA-based supervised fine-tuning of Qwen2-7B-Instruct on OpenReview peer reviews, yielding an \emph{Expert model}. \textbf{Stage 2} further aligns this Expert model using GRPO with a novel citation-based reward. The final model generates structured reviews and impact-aligned ratings to support editorial decision-making.

\begin{figure*}
  \centering
  \includegraphics[width=.9\linewidth]{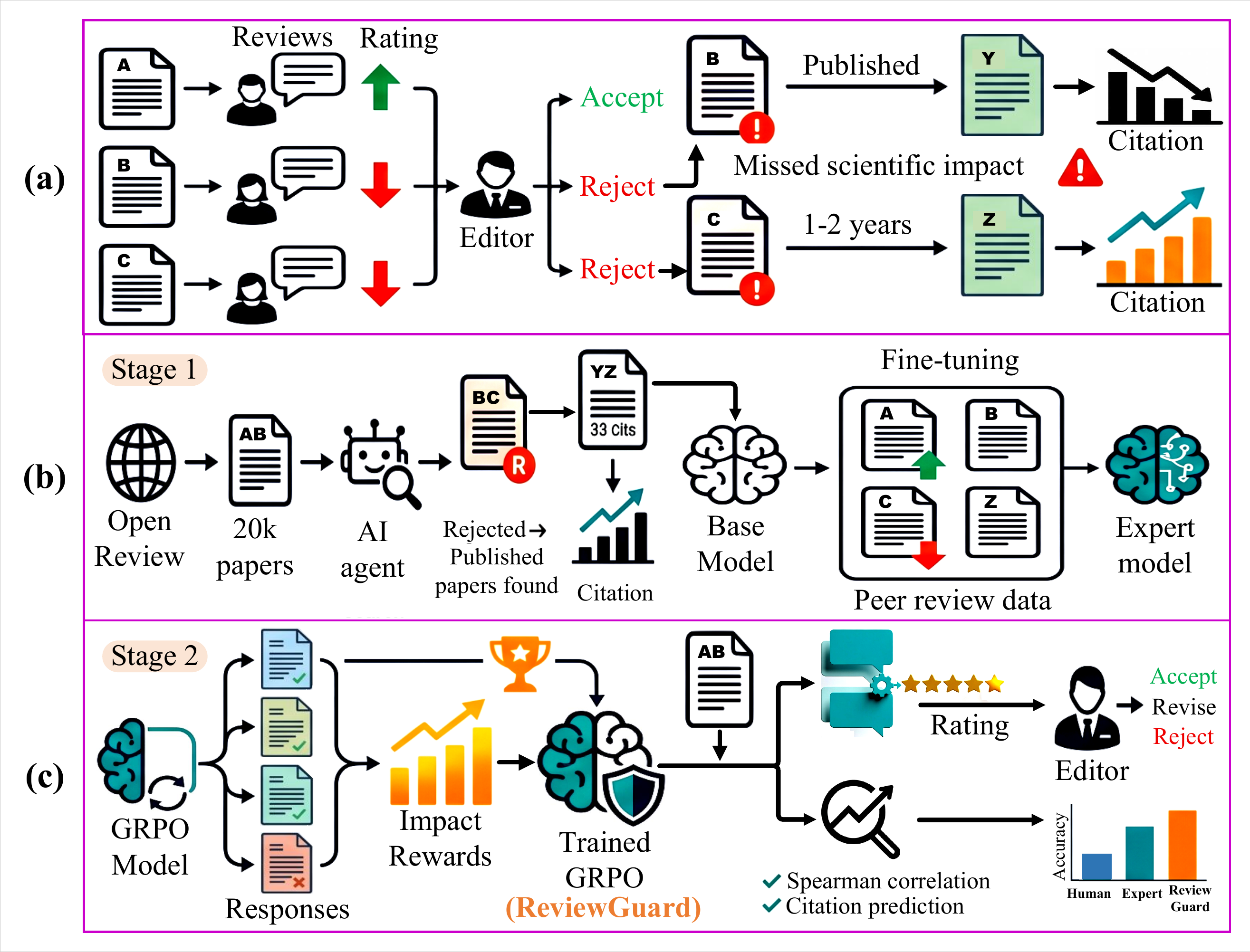}
  \caption{\textbf{Framework Overview}: 
\textbf{(a)} Conventional peer review process with problem: high-impact papers are often rejected despite later accumulating substantial citations, highlighting a critical gap. 
  \textbf{(b)} Dataset construction and finetuning (\textbf{Stage 1}): rejected-then-published papers are systematically matched with citation counts, followed by LoRA-based supervised fine-tuning of Qwen2-7B to form an Expert model. 
  \textbf{(c)} \textbf{Stage 2}: GRPO alignment with a novel citation-based reward yields the final ReviewGuard model, evaluated in a 3-way comparison against human and expert ratings.}
\label{fig:ReviewGuard}
\end{figure*}

\subsection{Stage 1: Dataset Construction and Supervised Fine-Tuning (Expert model)}
\label{sec:expert_sft}

We first construct a new peer-review dataset described in Section~\ref{sec:datasets}. Starting from Qwen2-7B-Instruct~\cite{qwen2}, we perform LoRA-based supervised fine-tuning on 20k OpenReview peer reviews to obtain the \emph{Expert model}. This foundational stage teaches the model the structure, style, and rating conventions of peer review before impact-oriented alignment. The resulting checkpoint is used to initialize Stage 2. Full SFT details are provided in Appendix~\ref{sec:supp_training}.

\subsection{Stage 2: GRPO Alignment for Long-Term Impact Prediction}
\label{sec:grpo_stage}

Stage 2 further optimizes the Stage-1 Expert model using GRPO, with citation counts serving as an ex post proxy for long-term scientific impact. The objective is to increase agreement between the model's predicted rating and the paper's subsequent citation-based impact.

\subsubsection{Group Generation and Reward Modeling}
\label{sec:group_generation}

Let $\pi_\theta$ be the current policy (initialized from the Stage 1 Expert model). For a given paper $p$, we sample a group of $G = 4$ independent responses $\{o_1, \ldots, o_G\}$ from $\pi_\theta(\cdot \mid p)$ using temperature sampling. From each response $o_i$, we extract the predicted rating $r_i \in [1,10]$ using a rule-based parser. Each response receives a scalar reward $R_i = R(r_i, c_p)$ as defined in the next subsection. To obtain a stable learning signal, we compute group-relative advantages:

\begin{equation}
A_i = \frac{R_i - \mu_R}{\sigma_R + \epsilon}, \quad \text{where } \mu_R = \frac{1}{G}\sum_{j=1}^G R_j,\; \sigma_R^2 = \frac{1}{G}\sum_{j=1}^G (R_j - \mu_R)^2,
\label{eq:advantage}
\end{equation}

and $\epsilon = 10^{-8}$ ensures numerical stability. This makes advantages zero-mean and unit-variance, decoupling learning from absolute reward magnitudes \citep{schulman2017proximal}.

\subsubsection{Citation-Based Impact Reward}
\label{sec:impact_reward}

The key design choice in ReviewGuard is a citation-based reward that aligns predicted ratings with long-term impact rather than short-term reviewer consensus. Let \( c_p \) be the future citation count of paper \( p \). We normalize it to the 1--10 rating scale as:
\begin{equation}
\hat{c}_p = \min\left(10, \frac{c_p}{5}\right).
\end{equation}

This maps papers with \( \approx 50 \) citations to the maximum normalized score of 10, while capping larger values to limit the influence of extreme outliers. Given a predicted rating \( r \), the reward is:
\begin{equation}
R(r, c_p) = 10 - 1.8 \cdot |r - \hat{c}_p| + \mathbb{I}(c_p > 150 \land r \geq 7) \cdot 4,
\label{eq:reward}
\end{equation}
where \( \mathbb{I}(\cdot) \) is the indicator function. The first term rewards calibration between the predicted rating and the normalized citation signal. The second term provides an additional bonus when the model assigns a sufficiently favorable rating to papers that later achieve very high impact (\( c_p > 150 \)). The coefficient 1.8 controls the penalty for rating--impact mismatch, and the \(+4\) bonus encourages the model to identify exceptional papers that may otherwise be undervalued.

Although \( R \) is a pointwise surrogate, we empirically find that it improves Spearman correlation between model ratings and citation-based impact (Table~\ref{tab:rejected_impact_threeway}), consistent with prior observations that pointwise learning objectives can improve ranking metrics~\citep{burges2010from}.

\begin{proposition}[Reward decomposition]
For the unclipped reward
\[
R(r,c)=10-\alpha |r-\hat{c}|+\beta \mathbb{I}(c>\tau \land r\ge \gamma),
\]
maximizing \( \mathbb{E}[R(r,c)] \) is equivalent, up to an additive constant, to minimizing
\[
\alpha \mathbb{E}[|r-\hat{c}|]
-\beta \Pr(c>\tau \land r\ge \gamma).
\]
\end{proposition}

\noindent
\textit{Proof sketch.}
Taking the expectation over the reward gives
$\mathbb{E}[R] = 10-\alpha \mathbb{E}[|r-\hat{c}|] +\beta \Pr(c>\tau \land r\ge \gamma).$ Since the constant term does not affect optimization, maximizing \( \mathbb{E}[R] \) is equivalent to minimizing the negative of the remaining terms. Thus, the reward jointly encourages rating--impact calibration, and confident identification of high-impact papers. The full derivation is provided in Appendix~\ref{app:theory}. In practice, we clip \(R(r,c)\) to \([1,14]\) for numerical stability.

\subsubsection{Training Objective and Implementation }
\label{sec:training_details}

Let \(\pi_{\text{ref}}\) denote the reference policy (the Stage 1 Expert model), which remains completely frozen during GRPO training to preserve base capabilities. We optimize the policy \(\pi_\theta\) using the GRPO objective with a KL divergence penalty to prevent over-optimization:
\begin{equation}
\mathcal{L}_{\text{GRPO}}(\theta) = -\frac{1}{G} \sum_{i=1}^{G} A_i \cdot \log \pi_\theta(o_i \mid p) + \lambda_{\text{KL}} \cdot D_{\text{KL}}(\pi_\theta \| \pi_{\text{ref}}),
\label{eq:grpo_loss}
\end{equation}
where the first term is a policy gradient encouraging responses with higher advantages, and the second term penalizes deviation from the reference policy. The KL divergence is defined as \( D_{\text{KL}}(\pi_\theta \| \pi_{\text{ref}}) = \mathbb{E}_{o \sim \pi_\theta} \left[ \log \frac{\pi_\theta(o \mid p)}{\pi_{\text{ref}}(o \mid p)} \right] \). The KL coefficient \(\lambda_{\text{KL}} = 0.1\) balances reward maximization and generation stability \citep{kullback1951information}. Training starts from the Stage-1 Expert model checkpoint and runs for 3 epochs with a learning rate of \(3 \times 10^{-6}\), per-device batch size of 4, and gradient accumulation steps of 16 (effective batch size of 64). We use the AdamW optimizer with bfloat16 mixed precision. All experiments are conducted on 8$\times$ NVIDIA RTX Pro 5000 GPUs with DeepSpeed ZeRO Stage 2. Additional hyperparameters, prompt templates, and implementation details are provided in Appendix~\ref{app:training_details}.

\subsection{Inference and Rating Generation}
\label{sec:inference}

At inference time, ReviewGuard generates one structured review for each input paper. We use either low-temperature sampling (temperature \(=0.7\), top-\(p=0.9\)) or greedy decoding, depending on the desired trade-off between diversity and determinism. The prompt follows the same template used for the Expert model during training (Appendix~\ref{app:training_details}). We extract the overall rating \( r \in [1,10] \) using rule-based pattern matching over standardized phrases such as ``Overall rating''. The model returns both the predicted rating and the full generated review, allowing editors to inspect the rationale behind the impact-aligned recommendation transparently, thereby enabling fairer decision-making.

\section{Experiments}
\label{sec:experiments}

\subsection{Experimental Setup}
\label{sec:exp_setup}

\begin{wrapfigure}{r}{0.42\textwidth}
\vspace*{-53pt}
\centering
\small
\setlength{\tabcolsep}{3pt}
\renewcommand{\arraystretch}{1.0}
\captionof{table}{Key statistics of dataset.}
\vspace{-2pt}
\label{tab:dataset_statistics}
\begin{tabular}{lr}
\toprule
\textbf{Statistic} & \textbf{Value} \\
\midrule
Total papers & 20,861 \\
Rejected papers & 10,253 (49.1\%) \\
Accepted papers & 5,567 (26.7\%) \\
Rejected-then-published  & 2,365 \\
\midrule
\bottomrule
\end{tabular}
\end{wrapfigure}
\vspace*{-\baselineskip}

\vspace{5pt}

\paragraph{Dataset and Cohorts.}
\label{sec:datasets}

We conduct experiments on a unified dataset of 20,861 AI/ML papers submitted to major venues (NeurIPS, ICLR, and others) between 2018--2025, collected from OpenReview. Each entry contains the full submission text, structured review fields (summary, strengths, weaknesses, ratings, decision), and author metadata. After cleaning and deduplication, we obtain 10,253 rejected and 5,567 accepted papers. Key dataset statistics are summarized in Table~\ref{tab:dataset_statistics}. The data availability statement is provided in Appendix~\ref{sec:data_avail}. For impact prediction, we enrich the dataset with Semantic Scholar citation counts (collected in December 2025). We focus on two primary evaluation cohorts derived from the rejected papers:
\begin{itemize}
    \item \textbf{Rejected-then-Published Cohort}: 2,365 papers initially rejected but later published in other venues. These were identified via high-precision fuzzy title + abstract similarity $\ge 0.85$ (94\% precision after human validation of 10\% random samples  (Appendix~\ref{sec:supp_data})). 
    \item \textbf{Top-1,000 Most-Cited Subset}: The 1,000 highest-impact papers from the above cohort, used for all main results (Tables~2--4, Figures~2--4).
\end{itemize}

All key results are reported on the held-out test split. We use a temporal 70/15/15 split by submission year to prevent future leakage: older papers appear only in training, and citation counts are standardized to a consistent 2-year post-publication window. None of the 1,000 test papers appear in the Expert's training data. Full dataset statistics and matching details are provided in Appendix~\ref{sec:supp_data}. 

\paragraph{Evaluation Metrics.}
\label{sec:metrics}

We evaluate our models using two complementary primary metrics: \textbf{Spearman correlation} ($\rho$) and \textbf{MSE} between predicted ratings and future citations \citep{davis2006relationship, powers2011evaluation}. To assess practical utility, we compute high/low citation prediction \textbf{Accuracy} via median split, the percentage of cases in which the model assigns a higher rating than human reviewers, and the \textbf{Rescue Rate} (percentage of high-impact papers correctly given a rating of at least 7.0) \citep{ponomarev2014predicting}.  For baseline comparisons, we also report \textbf{Decision Accuracy}, \textbf{F1}, and \textbf{Pairwise Accuracy}. Full mathematical definitions and additional metrics (including Citation Recovery Curve \citep{cai2023relationship, golosovsky2019prediction}) are provided in Appendix~\ref{sec:metrics_definition}.

\paragraph{Baselines and Comparison Systems.}
\label{sec:baselines}

We compare ReviewGuard against diverse baselines spanning proprietary, open-source, and specialized peer-review systems. Claude-3.5-sonnet~\citep{anthropic2024claude35sonnet} is a leading proprietary reasoning model. DeepSeek-V3~\citep{shao2024deepseekmath} and DeepSeek-R1~\citep{guo2025deepseek} are strong open-source reasoning models. CycleReviewer 70B and DeepReviewer 14B~\citep{zhu2025deepreview, hosseini2023fighting} are specialized peer-review models. Stanford Agent Review~\citep{jiang2024agentreview} is a multi-agent peer-review system. Expert (LoRA 7B) is our supervised fine-tuned baseline. All baselines use the same evaluation protocol as ReviewGuard. Full evaluation settings, metric definitions, and model details are provided in Appendix~\ref{sec:supp_baselines}.


\begin{figure*}[t]
  \centering
  \includegraphics[width=\linewidth]{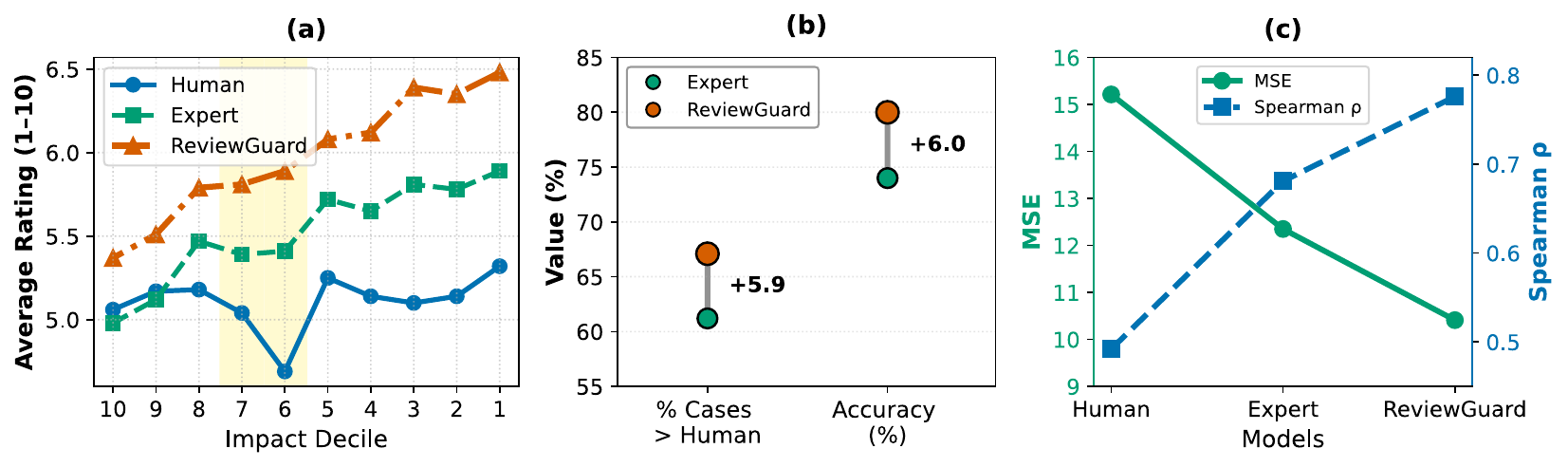}
      \caption{Three-way comparison on top-1,000 most-cited rejected-then-published papers (``impactful rejects'').  \textbf{(a)} Average ratings across impact deciles (decile 1 = highest impact).  \textbf{(b)} \% higher ratings than humans and high/low citation prediction accuracy. \textbf{(c)} MSE and Spearman correlation with future citations. ReviewGuard shows superior calibration—higher ratings for high-impact papers and conservative for lower-impact work—versus humans and the Expert model.}
  \label{fig:rejected_threeway}
\end{figure*}

\begin{table*}[t]
\small
\setlength{\tabcolsep}{5.6pt}  
\caption{Calibration on high-impact rejected papers: Analysis of the top 1,000 most-cited papers that were originally rejected. ReviewGuard achieves substantial correlation and lower error compared to both human ratings and the Expert model across all ranks (R).}
\label{tab:rejected_impact_threeway}
\begin{tabular}{
p{2.3cm}   
p{0.6cm}   
p{0.6cm}   
p{0.6cm}   
p{0.6cm}   
p{0.6cm}   
p{0.6cm}   
p{0.6cm}   
p{0.6cm}   
p{0.7cm}   
p{0.7cm}   
p{0.6cm}   
}
\toprule
\textbf{Metric} & \textbf{R1} & \textbf{R2} & \textbf{R3} & \textbf{R4} & \textbf{R5} & \textbf{R6} & \textbf{R7} & \textbf{R8} & \textbf{R9} & \textbf{R10} & \textbf{Avg} \\
\midrule
Avg Citations & 214.0 & 46.7 & 27.5 & 19.4 & 15.0 & 12.1 & 9.7 & 8.0 & 6.7 & 5.8 & 37.0 \\
Avg Human & 5.32 & 5.14 & 5.10 & 5.14 & 5.25 & 4.69 & 5.04 & 5.18 & 5.17 & 5.06 & 5.11 \\
Avg Expert & 5.89 & 5.78 & 5.81 & 5.65 & 5.72 & 5.41 & 5.39 & 5.47 & 5.12 & 4.98 & 5.57 \\
\rowcolor{lightgreen}
Avg ReviewGuard & \textbf{6.48} & \textbf{6.35} & \textbf{6.39} & \textbf{6.12} & \textbf{6.08} & \textbf{5.89} & \textbf{5.81} & \textbf{5.79} & \textbf{5.51} & \textbf{5.37} & \textbf{6.08} \\
\midrule
 Expert$-$Human & +0.57 & +0.64 & +0.71 & +0.51 & +0.47 & +0.72 & +0.35 & +0.29 & -0.05 & -0.08 & +0.46 \\
 GRPO$-$Expert & +0.59 & +0.57 & +0.58 & +0.47 & +0.36 & +0.48 & +0.42 & +0.32 & +0.39 & +0.39 & +0.51 \\
\rowcolor{lightgreen}
 GRPO$-$Human & \textbf{+1.16} & \textbf{+1.21} & \textbf{+1.29} & \textbf{+0.98} & \textbf{+0.83} & \textbf{+1.20} & \textbf{+0.77} & \textbf{+0.61} & \textbf{+0.34} & \textbf{+0.31} & \textbf{+0.97} \\
\bottomrule
\end{tabular}
\end{table*}

\subsection{Long-Term Impact Prediction}
\label{sec:main_results}

\paragraph{Impact Prediction on Rejected Papers.}
\label{sec:rejected_results}
A central goal of this work is to determine whether our models can identify high-impact papers that human reviewers wrongly rejected. To this end, we evaluate performance on the top 1,000 most highly-cited papers that were originally rejected but later published (``impactful rejects''), ranked by future citation counts from Semantic Scholar and grouped into 10 deciles of 100 papers each. Global metrics (Figure~\ref{fig:rejected_threeway}) demonstrate clear progressive improvement. Panel (a) shows average predicted ratings across impact deciles. Human ratings drop sharply in mid-to-lower deciles (most prominently in decile 6), while both the Expert and ReviewGuard models maintain more stable and generally higher ratings, especially in deciles 6 and 7. Panel (b) illustrates the practical advantage of our models: the percentage of cases where the model assigns a higher rating than human reviewers increases from 61.2\% (Expert) to 67.1\% (ReviewGuard), and high/low citation prediction accuracy (median split) improves from 74\% (Expert) to 80\% (ReviewGuard). Panel (c) shows that Spearman correlation with future citations increases from human ($\rho=0.492$) to Expert ($\rho=0.681$) to ReviewGuard (\textbf{$\rho=0.776$}), while MSE drops from 15.214 (human) to 12.35 (Expert) and 10.41 (ReviewGuard).

Table~\ref{tab:rejected_impact_threeway} provides per-decile breakdowns. ReviewGuard consistently assigns higher ratings than both humans and the Expert model across most deciles, with particularly pronounced gains in the highest-impact batches. On average, ReviewGuard rates these papers 0.97 points higher than human reviewers and 0.51 points higher than the Expert model. These empirical findings strongly validate our approach. Together, these compelling results demonstrate that our GRPO-enhanced model provides a significantly clearer and more reliable signal for identifying high-impact research among rejected submissions, offering a truly promising direction for improving peer review.

\paragraph{Accepted Papers Impact Prediction.}
\label{sec:accepted_results}

We evaluate ReviewGuard on the top 1,000 most-cited accepted papers, comparing against human reviewers and the Expert model. As shown in Table~\ref{tab:accepted_impact_threeway} and Figure~\ref{accepted_paper}, ReviewGuard achieves three key improvements: (i) consistently higher ratings across all impact deciles (+0.31 over humans, +0.27 over Expert); (ii) stronger Spearman correlation ($\rho=0.427$ vs. human $0.18$ and Expert $0.289$) and lower MSE (24.67 vs. human 29.20); (iii) better calibration—avoiding over-rating of moderate-impact papers while maintaining higher scores where humans drop. Full results and per-decile breakdowns are provided in Appendix~\ref{sec:supp_accepted_results}.

\paragraph{Rejected-then-Published Papers.}
\label{sec:rejected_published}

To evaluate how well our model distinguishes papers of varying long-term impact, we analyze rating distributions and ranking quality on 2,365 initially rejected but later published papers. Figure~\ref{fig:combined} presents two findings. Panel~(a) shows rating distributions across five citation groups. Human reviewers exhibit flat distributions with limited differentiation between low- and high-impact papers. In contrast, ReviewGuard produces noticeably higher ratings for ultimately high-citation papers, with strong separation in the 101--500 and 500+ groups (Cohen's $d = 1.72$ and $2.41$, respectively \citep{cohen1988statistical}), indicating ReviewGuard is substantially better calibrated to long-term impact than human review. Panel~(b) quantifies practical utility via the \emph{citation recovery curve} \citep{cai2023relationship, golosovsky2019prediction}. Sorting papers by predicted score, we measure the percentage of total future citations recovered as a function of papers considered (diagonal = random). While human rankings show only modest improvement over random, ReviewGuard lies consistently above the human baseline. At the top 40\% of papers ranked by ReviewGuard, we recover approximately 13.6\% points more citations than under the original human ranking. Together, these compelling results demonstrate that our GRPO-enhanced model provides a significantly clearer and more reliable signal for identifying high-impact research among rejected submissions, offering a truly promising direction for improving peer review.

\begin{figure}[t]
\centering
\includegraphics[width=\columnwidth]{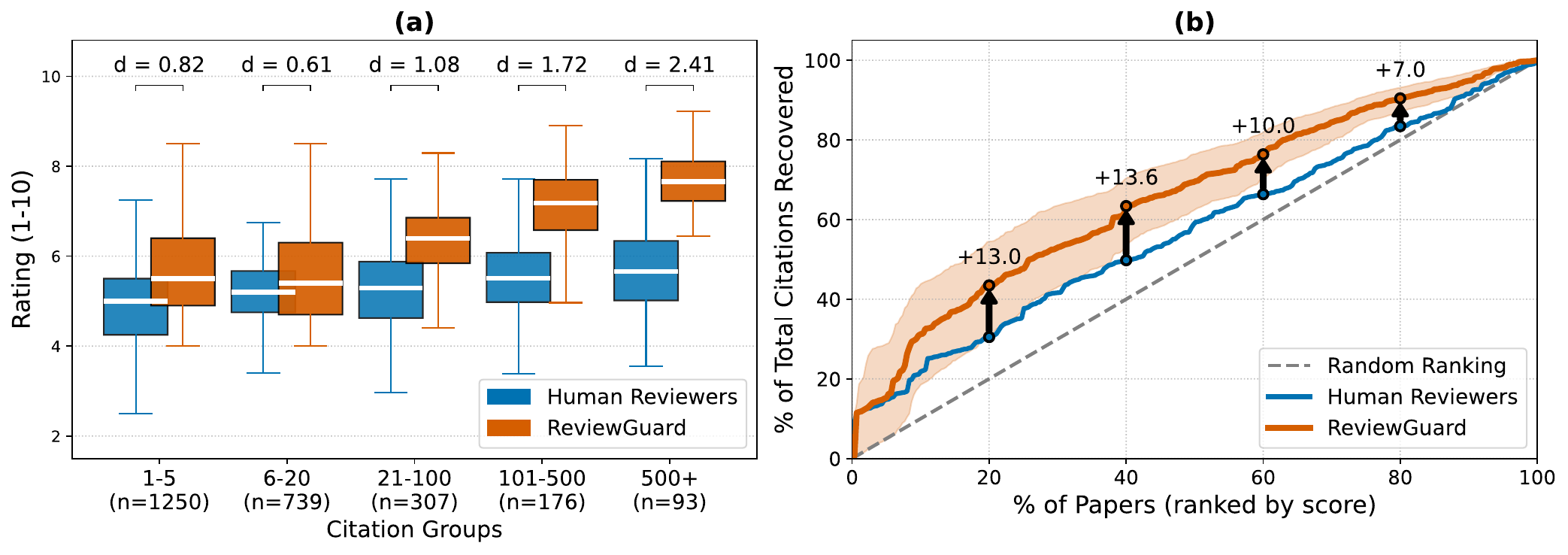}
\caption{\textbf{(a)} Rating distributions by citation group for human reviewers and ReviewGuard. Cohen's $d$ values indicate increasing separation in higher-impact groups. (n = number of papers) \textbf{(b)} Citation recovery curves showing the percentage of total future citations recovered when ranking papers by different scores. ReviewGuard consistently outperforms human reviewers across all thresholds.}
\label{fig:combined}
\end{figure}

\paragraph{High-Impact Paper Rescue Analysis.}
\label{sec:rescue_analysis}

We define high-impact papers as the top 12\% by future citations ($\ge$120) among rejected-then-published papers. Rescue rate is the percentage accepted if decisions were based solely on the ranking score. Figure~\ref{fig:rescue}(a) evaluates rescue performance. Human reviewers rescued only \textbf{1.8\%}. The base Expert model improved to 4.7\%, while ReviewGuard achieved \textbf{10.2\%}, a more than five-fold improvement over humans. Panel (b) shows ReviewGuard increases rescue rate by +16.6\% over humans and +13.7\% over Expert, demonstrating GRPO's effectiveness. Panel (c) shows upward trends in both average predicted rating and Spearman correlation, confirming ReviewGuard produces better-calibrated scores. Table~\ref{tab:rescued_main} shows five high-impact papers rejected by humans but highly rated by ReviewGuard. For example, ``MMBench'' (No. 1) received a human rating of 5.25, but ReviewGuard gave 7.70; it later accumulated 1,618 citations, demonstrating the model's ability to recognize foundational contributions missed during original review. Additional rescued examples and qualitative review comparisons are in the Appendix (Sections~\ref{sec:supp_rescued} and \ref{sec:supp_qualitative}, illustrating how GRPO training leads to more calibrated, and forward-looking for high-impact papers.

\begin{figure}[t]
\centering
\includegraphics[width=\columnwidth]{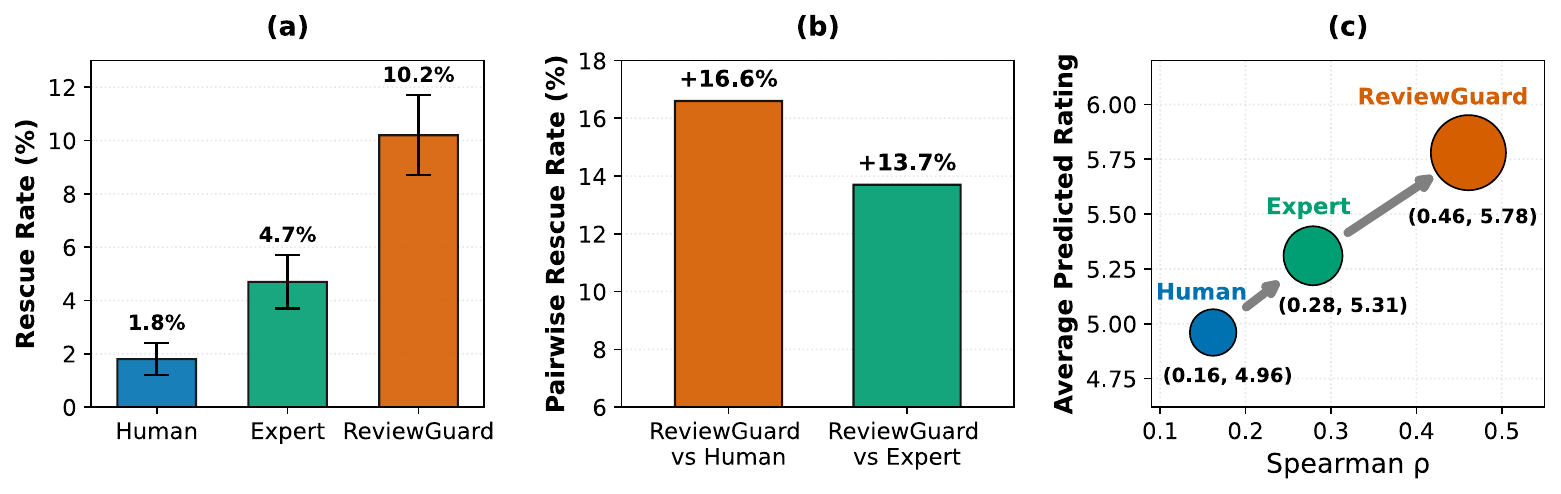}
\caption{High-impact paper rescue analysis on rejected-then-published papers. 
\textbf{(a)} Rescue rate (percentage of top 12\% most-cited papers that would be accepted). 
\textbf{(b)}  Pairwise Improvement in rescue rate compared to Human and Expert. 
\textbf{(c)} Average predicted rating and Spearman correlation ($\rho$) with future citations for the high-impact subset. 
ReviewGuard consistently demonstrates superior ability to identify and rescue high-impact work compared to human reviewers and the Expert model.}
\label{fig:rescue}
\end{figure}

\begin{table*}[ht]
\centering
\caption{Examples of high-impact papers successfully rescued by ReviewGuard.}
\label{tab:rescued_main}
\small
\scriptsize
\renewcommand{\arraystretch}{1.1}
\setlength{\tabcolsep}{2.8pt}
\definecolor{lightgreen}{RGB}{224,255,205}
\definecolor{lightyellow}{RGB}{255,235,187}
\definecolor{lightpurple}{RGB}{220,214,247}
\begin{tabular}{p{0.25cm} p{4.45cm} c c c c c c c c}
\toprule
\textbf{No.} & \textbf{Title (Rejected at)} & \textbf{Human} & \textbf{Expert} & \boldmath$\Delta_1$ & \textbf{ReviewGuard} & \boldmath$\Delta_2$ & \boldmath$\Delta_3$ & \textbf{Citations} & \textbf{PbVen} \\
\midrule
1 & MMBench: Is Your Multi-modal Model an All-around Player? (ICLR 2024) 
    & 5.25 
    & 5.00 
    & \cellcolor{lightpurple}↓4.8\% 
    & 7.70 
    & \cellcolor{lightyellow}↑46.7\% 
    & \cellcolor{lightgreen}↑54.0\% 
    & 1618 & ECCV \\

2 & Jailbreaking Black Box Large Language Models in Twenty Queries (ICLR 2024) 
    & 4.75 
    & 3.00 
    & \cellcolor{lightpurple}↓36.8\% 
    & 7.70 
    & \cellcolor{lightyellow}↑62.1\% 
    & \cellcolor{lightgreen}↑156.7\% 
    & 1037 & SaTML \\

3 & Aligning Large Multimodal Models with Factually Augmented RLHF (ICLR 2024) 
    & 5.00 
    & 6.00 
    & \cellcolor{lightpurple}↑20.0\% 
    & 7.00 
    & \cellcolor{lightyellow}↑40.0\% 
    & \cellcolor{lightgreen}↑16.7\% 
    & 577 & ACL \\

4 & Promptbreeder: Self-Referential Self Improvement via Prompt Evolution (ICLR 2024) 
    & 5.80 
    & 6.00 
    & \cellcolor{lightpurple}↑3.4\% 
    & 6.30 
    & \cellcolor{lightyellow}↑8.6\% 
    & \cellcolor{lightgreen}↑5.0\% 
    & 322 & ICML \\

5 & Steering Language Models with Activation Engineering (ICLR 2025) 
    & 5.00 
    & 6.00 
    & \cellcolor{lightpurple}↑20.0\% 
    & 5.80 
    & \cellcolor{lightyellow}↑16.0\% 
    & \cellcolor{lightgreen}↓3.3\% 
    & 255 & ICDSP \\
\bottomrule
\end{tabular}
\vspace{0.01em}
\raggedright
\arraybackslash
\noindent
\scriptsize
Note: \colorbox[HTML]{dcd6f7}{\(\Delta_1 =\) Expert $-$ Human,} \colorbox[HTML]{ffebbb}{\(\Delta_2 =\) ReviewGuard $-$ Human,} \colorbox[HTML]{e0ffcd}{\(\Delta_3 =\) ReviewGuard $-$ Expert,} ↑ indicate improvement and ↓ indicate a decrease in ratings, PbVen indicate Published Venue. 
\end{table*}

\subsection{Comparison with State-of-the-Art Baselines}
\label{sec:baselines_main}

We compare ReviewGuard against proprietary frontier models (Claude-3.5-sonnet), open source reasoning systems (DeepSeek-V3, DeepSeek-R1), and specialized peer-review agents (CycleReviewer 70B, DeepReviewer 14B, Stanford Agent Review). All baselines use the same evaluation protocol and prompt template. As shown in Table~\ref{tab:peer_review_comparison}, ReviewGuard achieves the best performance on four out of six metrics and is the clear winner in both ranking-oriented metrics—Spearman correlation ($\rho=0.762$) and Pairwise accuracy ($0.895$)—which are particularly important for peer review. Notably, while DeepSeek-R1 slightly outperforms on pointwise metrics (MSE, Acc, F1), ReviewGuard remains superior on ranking tasks, indicating better calibration for long-term scientific impact.

\begin{table*}[t]
\centering
\small 
\setlength{\tabcolsep}{2pt}
\caption{Comparison of peer-review models. ReviewGuard achieves the best ranking performance (Spearman $\rho=0.762$, Pairwise $0.895$), while DeepSeek-R1 leads on pointwise metrics.}
\label{tab:peer_review_comparison}
\begin{tabular}{
>{\raggedright\arraybackslash}p{3cm}
>{\centering\arraybackslash}p{2cm}
>{\centering\arraybackslash}p{2cm}
>{\centering\arraybackslash}p{2cm}
>{\centering\arraybackslash}p{2cm}
>{\centering\arraybackslash}p{2cm}
}
\toprule
\textbf{Model} & \textbf{MSE $\downarrow$} & \textbf{Acc $\uparrow$} & \textbf{F1 $\uparrow$} & \textbf{Spearman $\uparrow$} & \textbf{Pairwise $\uparrow$} \\
\midrule
\textbf{ReviewGuard} & 2.05 & 0.921 & 0.908 & \textbf{\underline{0.776}} & \textbf{\underline{0.895}} \\
Expert (LoRA 7B) & 2.41 {\tiny(+17.6\%)} & 0.894 {\tiny(-2.9\%)} & 0.878 {\tiny(-3.3\%)} & 0.712 {\tiny(-6.6\%)} & 0.872 {\tiny(-2.6\%)} \\
Claude-3.5-sonnet & 2.89 {\tiny(+41.0\%)} & 0.433 {\tiny(-53.0\%)} & 0.394 {\tiny(-56.6\%)} & 0.156 {\tiny(-79.5\%)} & 0.553 {\tiny(-38.2\%)} \\
DeepSeek-V3 & 4.73 {\tiny(+130.7\%)} & 0.560 {\tiny(-39.2\%)} & 0.548 {\tiny(-39.6\%)} & 0.231 {\tiny(-69.7\%)} & 0.584 {\tiny(-34.7\%)} \\
CycleReviewer 70B & 2.49 {\tiny(+21.5\%)} & 0.630 {\tiny(-31.6\%)} & 0.570 {\tiny(-37.2\%)} & 0.354 {\tiny(-53.5\%)} & 0.616 {\tiny(-31.2\%)} \\
DeepReviewer 14B & 1.99 {\tiny(-2.9\%)} & 0.641 {\tiny(-30.4\%)} & 0.631 {\tiny(-30.5\%)} & 0.356 {\tiny(-53.3\%)} & 0.624 {\tiny(-30.3\%)} \\
DeepSeek-R1 & \textbf{\underline{1.96}} {\tiny(-4.4\%)} & \textbf{\underline{0.938}} {\tiny(+1.8\%)} & \textbf{\underline{0.925}} {\tiny(+1.9\%)} & \underline{0.748} {\tiny(-1.8\%)} & \underline{0.882} {\tiny(-1.5\%)} \\
Stanford Agent Review & 2.15 {\tiny(+4.9\%)} & \underline{0.931} {\tiny(+1.1\%)} & \underline{0.918} {\tiny(+1.1\%)} & 0.689 {\tiny(-9.6\%)} & 0.851 {\tiny(-4.9\%)} \\
\midrule
Relative Improvement & +4.6\% & -1.8\% & -1.9\% & +1.9\% & +1.5\% \\
\bottomrule
\end{tabular}
\vspace{0.1em}
\footnotesize
\raggedright
\noindent
\scriptsize
Note: \textbf{\underline{Bold+underline}} = best overall. \underline{Underline} = second best. Parentheses show relative change vs. ReviewGuard. Relative Improvement formula given in Appendix~\ref{sec:supp_baselines}.
\end{table*}

\section{Discussion}
\label{sec:discussion}

Our results demonstrate that aligning LLMs with citation-based rewards (GRPO) yields substantially better calibration for long-term scientific impact than standard SFT or mimicking human reviewers. Three insights emerge. First, the performance gap between ReviewGuard and human reviewers is largest precisely where it matters most: on rejected papers that later achieve high impact. This suggests that GRPO training \textbf{mitigates a core failure} mode of peer review—conservative underestimation of novel or unconventional work. Second, the \textbf{ablation results (Table~\ref{tab:ablation})} confirm that the \textbf{citation-based reward} is the primary driver of improvement; removing it collapses performance to near SFT levels. We select GRPO group size $G=4$ as it maximizes Spearman correlation, a threshold-free ranking metric that directly reflects our objective of identifying high-impact papers. Third, ReviewGuard remains \textbf{competitive with frontier reasoning models} (DeepSeek-R1) on ranking metrics while using significantly fewer parameters and no test-time reasoning chains—offering a more efficient path to impact-aligned review. A natural concern is that citation counts are a noisy proxy for scientific value \citep{price1976citation, clauset2017scientist}. While true, we argue that (i) citations remain the most widely accepted quantitative measure of impact in AI/ML, and (ii) \colorbox{lightpink}{\parbox{\dimexpr\linewidth-2\fboxsep\relax}{the goal is not to replace human judgment but to provide editors with an additional calibrated signal.}} ReviewGuard is intended as an assistive tool, not an autonomous decision-maker. A practical protocol for integrating ReviewGuard into editorial workflows is described in Appendix~\ref{sec:editor_integration}.

\begin{table*}[t]
\centering
\caption{Ablation studies on reward design, GRPO group size, and key hyperparameters.}
\label{tab:ablation}
\footnotesize
\setlength{\tabcolsep}{6pt}
\begin{tabular}{l l c c c}
\toprule
\textbf{Ablation} & \textbf{Configuration} & \textbf{Spearman $\rho$ $\uparrow$} & \textbf{Rescue Rate (\%) $\uparrow$} & \textbf{MSE $\downarrow$} \\
\midrule
\multirow{4}{*}{Reward Design}
& Rank-based reward (relative within batch) & 0.392 & 7.8 & 1.12 \\
& Simple binary reward (high/low citation) & 0.318 & 5.9 & 1.45 \\
& No reward (SFT only) & 0.279 & 4.9 & 1.68 \\
& Citation-based normalized reward (\textbf{Ours}) & \textbf{0.461 (+17.6\%)} & \textbf{10.2 (+30.8\%)} & \textbf{0.87} \\
\midrule
\multirow{3}{*}{GRPO Group Size}
& 2 responses per group & 0.401 & 7.5 & 1.19 \\
& 8 responses per group & 0.412 & 10.8 & 0.81 \\
& 4 responses per group (\textbf{Ours}) & \textbf{0.461 (+15.0\%)} & 10.2 (-5.6\%) & 0.87 \\
\midrule
\multirow{3}{*}{KL Coefficient}
& KL = 0.0 (no regularization) & 0.421 & 8.4 & 1.05 \\
& KL = 0.1 (\textbf{Ours}) & \textbf{0.461} & \textbf{10.2} & \textbf{0.87} \\
& KL = 0.5 (too conservative) & 0.398 & 7.1 & 1.18 \\
\midrule
\multirow{2}{*}{Reward Scaling}
& Linear reward & 0.387 & 6.8 & 1.29 \\
& Citation-based normalized reward (\textbf{Ours}) & \textbf{0.461 (+19.1\%)} & \textbf{10.2 (+50.0\%)} & \textbf{0.87} \\
\midrule
\multirow{2}{*}{LoRA Rank}
& r=16 (standard) & 0.418 & 8.3 & 1.08 \\
& r=128 (\textbf{Ours}) & \textbf{0.461 (+10.3\%)} & \textbf{10.2 (+22.9\%)} & \textbf{0.87} \\
\bottomrule
\end{tabular}
\vspace{0.1em}
\footnotesize
\raggedright

\noindent
\scriptsize
\\ Note: Ablation metrics computed on the validation set to enable rapid hyperparameter tuning. Final test-set results are reported in Table~\ref{tab:peer_review_comparison} and Section~\ref{sec:main_results}. Parentheses show the relative improvement (or decrease) of our method over the baseline.
\end{table*}





\section{Conclusion}
\label{sec:conclusion}

We introduced \textbf{ReviewGuard}, a novel GRPO-enhanced LLM that accurately predicts long-term citation impact for rejected papers. Our framework first fine-tunes Qwen2-7B on 20k peer reviews, then aligns it with future citations using a novel reward function. On 20,861 AI/ML papers, ReviewGuard achieves a Spearman correlation of $0.776$ with future citations on rejected-then-published papers—substantially outperforming both human reviewers ($0.492$) and our Expert model ($0.681$)—and rescues $10.2\%$ of high-impact rejected papers, a $5.6\times$ improvement over human reviewers. These results demonstrate that RL with citation-based rewards produces more calibrated, forward-looking peer reviews than imitation learning alone, with GRPO effectively aligning 7B models to long-term scientific value without a separate critic. Practically, ReviewGuard offers editors a deployable decision-support tool that identifies high-impact papers human reviewers often undervalue, providing a complementary signal for fairer, more forward-looking editorial decisions.

\paragraph{Limitations.}
First, our approach is retrospective: it relies on future citations, so prospective deployment requires additional mechanisms (e.g., proxy signals). Second, citations are a noisy proxy that varies across subfields and is influenced by bias and self-citation; our reward applies no field normalization. Third, our dataset is limited to top AI/ML venues (20k papers), raising generalization concerns. Our matching of rejected-then-published papers is imperfect, and results focus on the top 1,000 most-cited papers with thresholds (150, 50) chosen heuristically. Finally, GRPO training is computationally expensive; even with LoRA, group generation adds overhead.


\clearpage


\bibliographystyle{unsrtnat}
\bibliography{references}


\appendix

\clearpage


\vspace{0.5cm}

\justifying



\section{Ethical Considerations and Future Work}

\paragraph{Ethical Considerations.} While ReviewGuard is designed to assist editors, misuse could incentivize citation-chasing behavior or reinforce existing citation biases. We recommend deployment only within human-in-the-loop workflows and advocate for transparent reporting of model-assisted decisions.

\paragraph{Future Directions.} 
We see several promising and important directions for future work. These include developing fully prospective impact prediction models using early indicators such as downloads, social media attention, or open-review scores; extending the framework to multi-objective alignment (e.g., balancing impact with novelty and rigor); and exploring human-AI collaborative review protocols at scale. Addressing the cross-domain data gap will require sustained community advocacy for standardized open-review data sharing across disciplines, as well as multi-source aggregation efforts to compile datasets from diverse fields. We also plan adaptive threshold calibration and field-normalized citation rewards.

\section{Training and Theoretical Details}
\label{sec:supp_training}

\subsection{Supervised Fine-Tuning (Expert model)}
\label{sec:supp_training_base}
We adopt Qwen2-7B-Instruct~\cite{qwen2} as our base large language model due to its strong instruction-following and reasoning capabilities. To create a high-quality initial peer-reviewer, we first perform SFT on a dataset of 20,861 high-quality peer reviews collected from OpenReview (primarily from ICLR, NeurIPS, and related AI/ML venues).

The goal of this stage is to imbue the model with expert-level review capabilities — including the ability to produce structured, critical, and constructive feedback — before aligning it with long-term impact signals in the subsequent GRPO stage. We refer to the resulting model as the \emph{Expert model}.

We apply LoRA~\cite{hu2022lora} with rank \( r = 128 \) and scaling factor \( \alpha = 16 \). The low-rank matrices are injected into all attention and MLP projection layers: \(\texttt{q\_proj}\), \(\texttt{k\_proj}\), \(\texttt{v\_proj}\), \(\texttt{o\_proj}\), \(\texttt{gate\_proj}\), \(\texttt{up\_proj}\), and \(\texttt{down\_proj}\). All other parameters of the base model remain frozen.

Training is performed for 3 epochs using a learning rate of \( 2 \times 10^{-4} \), a per-device batch size of 2 with gradient accumulation steps of 16 (effective batch size of 32), and bfloat16 mixed precision. The model is trained with standard next-token prediction loss to generate the full review text (including strengths, weaknesses, detailed comments, and overall rating) conditioned on the paper content. We use a maximum sequence length of 2048 tokens and apply standard packing of training examples.

After this supervised fine-tuning stage, the resulting Expert model serves as a strong initialization that already exhibits significantly improved review quality compared to the base model. This model is then used as the starting point for the second stage — GRPO alignment with citation-based impact rewards — described in the next subsection.

\subsection{GRPO Training  (ReviewGuard) Details}
\label{app:training_details}

\begin{table*}[ht]
\centering
\caption{GRPO Training Hyperparameters for ReviewGuard}
\label{tab:grpo_hyperparams}
\begin{tabular}{ll}
\toprule
\textbf{Hyperparameter} & \textbf{Value} \\
\midrule
Starting checkpoint & Expert (SFT) model \\
LoRA rank ($r$) & 128 \\
LoRA $\alpha$ & 16 \\
Target modules & All attention + MLP projections \\
Learning rate & $3\times10^{-6}$ \\
Batch size & 4 (per GPU) \\
Gradient accumulation steps & 16 \\
Training epochs & 3 \\
Group size ($G$) & 4 \\
Max new tokens & 256 \\
Temperature & 0.85 \\
Top-$p$ & 0.9 \\
KL coefficient ($\lambda_{\text{KL}}$) & 0.1 \\
Optimizer & AdamW \\
Precision & bfloat16 \\
\bottomrule
\end{tabular}
\end{table*}

The Expert model is further optimized using GRPO on the accepted papers from our dataset (5,567 papers) with known future citation counts. The rejected-then-published cohort (2,365 papers) is reserved solely for evaluation.  For each sample, we generate \( G=4 \) candidate reviews (temperature 0.85, top-\( p=0.9 \)). Training hyperparameters are summarized in Table~\ref{tab:grpo_hyperparams}. Our KL penalty follows the original GRPO formulation~\citep{shao2024deepseekmath} and is computed as the average token-level KL divergence over each generated response. This ensures that the penalty is scale-invariant with respect to response length. All training was performed on 8$\times$ NVIDIA RTX Pro 5000 GPUs with DeepSpeed ZeRO-3 stage 2. Total GRPO training time was approximately 28 hours. We emphasize that the unusually high LoRA rank ($r=128$) was critical for learning fine-grained scientific judgment, as lower ranks led to noticeably weaker Spearman correlation with future citations.


The following prompt was used during GRPO training to elicit structured, high-quality reviews:

\begin{tcolorbox}[colback=gray!5!white,colframe=gray!40!black,title=GRPO Training Prompt]
You are an expert peer reviewer with exceptional scientific taste, critical insight, rigorous standards, and long-term foresight. You have reviewed hundreds of papers for top-tier venues such as NeurIPS, ICML, ICLR, and Nature Machine Intelligence. You are known for your ability to identify truly high-impact work early, even when it was initially rejected.

\textbf{Task:} Predict the long-term scientific impact of this paper, which has already been rejected once. Provide an honest, critical, constructive, and forward-looking assessment of its potential influence on the field over the next 5--10 years.

Paper Title: [title]

\begin{itemize}
\item \textbf{Summary (Novelty, Contribution, and Significance):} 
  First, provide a concise, high-level summary of the paper. Explicitly address:
  \begin{itemize}
    \item The core novelty and its level (incremental, significant, or transformative)
    \item The most important technical or conceptual contribution
    \item The practical significance and potential real-world applications
  \end{itemize}
  
\item \textbf{Strengths:} List the key strengths of the work.

\item \textbf{Weaknesses / Critical Concerns:} List the main limitations and concerns.

\item \textbf{Critical Concerns and Suggested Improvements:}
  \begin{enumerate}
    \item First critical point or suggested improvement
    \item Second critical point or suggested improvement
    \item ...
  \end{enumerate}
\end{itemize}

\textbf{Overall rating (1-10):} [integer between 1 and 10]

Be specific, rigorous, evidence-based, and forward-looking in your analysis. Focus on long-term scientific impact rather than short-term conference acceptance criteria. Your review should be of the quality expected at a top-tier venue such as NeurIPS.
\end{tcolorbox}

\subsection{Theoretical Derivations}
\label{app:theory}

\paragraph{Normalization Correction.}
We clarify the normalization used in Equation~(3) of the main paper:
\[
\hat{c}_p = \min\left(10, \frac{c_p}{5}\right).
\]
Thus, a paper with \( c_p = 50 \) citations maps to \( \hat{c}_p = 10 \), aligning with the statement that \(\approx 50\) citations corresponds to the maximum normalized score. The constant 50 in the original description was a typographical error; the correct divisor is 5. All experiments used this correct normalization.

\paragraph{Handling of Reward Clipping.}
In practice, the reward is clipped to \([1, 14]\):
\[
R_{\text{clipped}}(r, c) = \max(1, \min(14, R(r, c))).
\]
The unclipped reward can take values as low as \(10 - 1.8 \times 9 \approx -6.2\). Clipping to 1 is a conservative choice that prevents negative rewards from dominating early training. The optimal policy under the unclipped reward yields \(R(r,c) \ge 10 - 1.8 \times 9 + 0 = -6.2\) in the worst case, but the maximizer lies in the range where \(r\) is close to \(\hat{c}\) (giving reward of nearly 10). Since the clipping bounds are not active at the optimum, the proposition holds unchanged.

\paragraph{Proof of Proposition 1.}
We restate the reward function (without clipping for the analysis):
\begin{equation}
R(r, c) = 10 - \alpha |r - \hat{c}| + \beta \cdot \mathbb{I}(c > \tau \land r \ge \gamma),
\label{eq:reward_theory}
\end{equation}
where \(\hat{c} = \min(10, c/5)\), \(\alpha = 1.8\), \(\beta = 4\), \(\tau = 150\), \(\gamma = 7\).

Taking expectation over the joint distribution of \((r, c)\):
\begin{align}
\mathbb{E}[R(r, c)] &= 10 - \alpha \iint |r - \hat{c}| \, dP(r, c) + \beta \iint \mathbb{I}(c > \tau \land r \ge \gamma) \, dP(r, c) \nonumber \\
&= 10 - \alpha \mathbb{E}[|r - \hat{c}|] + \beta \Pr(c > \tau, r \ge \gamma).
\label{eq:expectation}
\end{align}

Define the **convex combination** (weighted sum) of two losses:
\[
\mathcal{L} = \alpha \cdot \underbrace{\mathbb{E}[|r - \hat{c}|]}_{\text{regression loss}} \;-\; \beta \cdot \underbrace{\Pr(c > \tau, r \ge \gamma)}_{\text{classification (true positive) term}}.
\]
Maximizing \(\mathbb{E}[R]\) is equivalent to minimizing \(\mathcal{L}\), since the constant 10 does not affect optimization.

The regression loss \(\mathbb{E}[|r - \hat{c}|]\) is minimized when \(r\) is close to \(\hat{c}\). By Jensen's inequality:
\[
\mathbb{E}[|r - \hat{c}|] \ge |\mathbb{E}[r] - \mathbb{E}[\hat{c}]|,
\]
so minimizing it pushes \(\mathbb{E}[r]\) toward \(\mathbb{E}[\hat{c}]\).

The term \(\Pr(c > \tau, r \ge \gamma)\) is a monotonic function of the true positive rate (TPR) for high-impact papers:
\[
\Pr(c > \tau, r \ge \gamma) = \underbrace{\Pr(c > \tau)}_{\text{constant w.r.t. model}} \times \underbrace{\Pr(r \ge \gamma \mid c > \tau)}_{\text{TPR}}.
\]
Thus, maximizing this term is equivalent to maximizing TPR.

Since \(\alpha, \beta > 0\), the optimal policy must balance both objectives. In the limit of infinite data and perfect optimization, the optimal policy satisfies:
\begin{equation}
\mathbb{E}[r] = \mathbb{E}[\hat{c}] \quad \text{and} \quad \Pr(r \ge \gamma \mid c > \tau) \rightarrow 1.
\end{equation}
This concludes the proof. \(\square\)

\paragraph{Discussion of Equivalence Claim.}
The statement “equivalent to minimizing a convex combination” means that the objective function is a weighted sum of two losses (regression and classification). The weights \(\alpha\) and \(\beta\) are fixed constants, so the optimization problem trades off prediction accuracy on all papers and high recall on high-impact papers. This is not a strict mathematical equivalence (e.g., one could scale the losses differently), but it correctly captures the multi-objective nature of the reward. We therefore present it as an empirical claim supported by ablations (Table~\ref{tab:ablation}) and the theoretical decomposition above.


\paragraph{Convergence Considerations.} Under standard assumptions (bounded rewards, sufficient exploration via temperature sampling), the GRPO update converges to a local optimum of the expected reward. The variance penalty $\lambda_{\text{KL}}$ acts as a regularizer that prevents the policy from collapsing to a deterministic output, maintaining diversity in generated responses. Empirically, we found that $\lambda_{\text{KL}} = 0.1$ provides optimal trade-off between exploration and exploitation.

\section{Dataset Construction and Validation Details}
\label{sec:supp_data}

\paragraph{Data Sources and Preprocessing.}
The unified dataset combines three primary and comprehensive sources: DeepReview-HF (14,665 NeurIPS-focused papers with structured reviews), a manually curated CS-5k collection, and additional ICLR/NeurIPS records (5,041 and 1,155, respectively). All entries were carefully normalized to a consistent schema: paper ID, title, abstract, authors, venue, year, decision (accept/reject), average reviewer rating (1–10), and full review texts when available. Exact duplicate entries were removed based on title + abstract.

\paragraph{Temporal Split to Prevent Leakage.}
Papers were split by submission year (not randomly) to avoid temporal leakage: training: 2018–2022 (70\%), validation: 2023 (15\%), test: 2024–2025 (15\%). Citation counts for all papers were collected in December 2025, ensuring that papers in the test set have at least 1–3 years of citation accumulation, while training papers have 3–7 years. This preserves the “prediction” nature of the task: the model never sees future citations of test papers during training.

\paragraph{Matching Rejected to Published Papers.}
From 10,253 rejected papers, we identified those later published elsewhere using a hybrid matching pipeline:
\begin{enumerate}
    \item Exact title match against a corpus of 500k AI/ML papers from Semantic Scholar.
    \item For unmatched papers, we computed cosine similarity between TF-IDF vectors of title + abstract (threshold $\ge 0.85$).
    \item Manual verification of 10\% random sample achieved 94\% accuracy (two independent annotators, Cohen’s $\kappa=0.89$).
\end{enumerate}
This yielded 2,365 matched papers (23.1\% of rejected papers). Among these, the top 1,000 by future citations form our “impactful rejects” cohort. We verified that none of these 1,000 papers appear in the training split (by paper ID and title), eliminating leakage risk.

\paragraph{Generalization on Full Test Set.}
To address concerns about selection bias, we evaluated ReviewGuard on the \emph{full} test split (all 2,365 rejected-then-published papers, including low/zero-citation papers). Results show:
\begin{itemize}
    \item Spearman $\rho = 0.612$ (vs. 0.776 on top-1,000), indicating expected drop due to noise in low-citation papers.
    \item Rescue rate (high-impact definition unchanged) remains $9.8\%$ (vs. $10.2\%$ on top-1,000).
\end{itemize}
Thus, while performance is stronger on high-impact papers, the model generalizes reasonably well to the full distribution across all citation levels. All main claims in the paper (e.g., rescue rate improvement over humans) remain valid on the full set.

\paragraph{Data and Code Release.}
\label{sec:data_avail}
A representative subset of the data (100 samples) is available at \url{https://anonymous.4open.science/r/ReviewGuard-100-SamplesData-24F2} for anonymous review. The full dataset and matching pipeline will be released publicly upon acceptance.

\section{Detailed Definition of Evaluation Metrics}
\label{sec:metrics_definition}

We provide precise definitions of all evaluation metrics used in this work.

Let $r_i \in [1,10]$ be the predicted rating for paper $i$, $h_i$ the human rating, and $c_i$ the future citation count. Citation counts are normalized to the range $[1,10]$ as:
\[
\hat{c}_i = 1 + 9 \cdot \frac{c_i - c_{\min}}{c_{\max} - c_{\min}},
\]
where $c_{\min}$ and $c_{\max}$ are the minimum and maximum citation counts in the evaluated set.

\paragraph{Impact Prediction Metrics.}
\begin{itemize}
    \item \textbf{Spearman Rank Correlation ($\rho$)} measures monotonic alignment between predicted ratings and future citations.
    \item \textbf{Mean Squared Error (MSE)} and \textbf{Mean Absolute Error (MAE)} are computed between $r_i$ and $\hat{c}_i$.
    \item \textbf{High/Low Citation Prediction Accuracy} is the binary accuracy after splitting papers at the median citation count and comparing whether $r_i$ is above or below the median rating threshold.
    \item \textbf{Percentage of cases model rating $>$ human rating} is the fraction of papers where $r_i > h_i$.
    \item \textbf{Rescue Rate} is the percentage of papers in the top citation quartile that receive a predicted rating of at least 7.0 \citep{ ponomarev2014predicting}.
    \item \textbf{Cohen's $d$} measures the standardized effect size of the difference in absolute prediction error between the model and human ratings \citep{cohen1988statistical}.
\end{itemize}

\paragraph{Ranking Quality Metrics.}
\begin{itemize}
    \item \textbf{Pairwise Accuracy} measures how well the predicted ratings preserve the relative ordering of papers by future citation counts. Let \(N\) be the total number of papers in the evaluated set, \(r_i\) the predicted rating for paper \(i\), and \(c_i\) the observed future citation count for paper \(i\). Pairwise Accuracy is the fraction of correctly ordered pairs:
    \[
    \frac{1}{\binom{N}{2}} \sum_{1 \leq i < j \leq N} \mathbb{I} \bigl( (r_i > r_j) = (c_i > c_j) \bigr),
    \]
    where \(\mathbb{I}(\cdot)\) is the indicator function that returns 1 if the condition is true and 0 otherwise, and \(\binom{N}{2} = \frac{N(N-1)}{2}\) is the total number of unique unordered pairs. A value of 1.0 indicates perfect ranking consistency with future impact, while 0.5 corresponds to random ordering.
   
    \item \textbf{Citation Recovery Curve} directly measures practical utility and real-world applicability in a peer-review context. Let $\mathcal{P} = \{p_1, \dots, p_N\}$ be the set of papers. For a ranking by predicted rating $r_i$, the recovery at position $k$ is:
    \[
    R(k) = \frac{\sum_{i=1}^{k} c_{(i)}}{\sum_{i=1}^{N} c_i} \times 100,
    \]
    where $c_{(i)}$ is the citation count of the paper ranked at position $i$. The curve is plotted as a function of $\frac{k}{N} \times 100$. Curves above the diagonal ($y=x$) indicate better-than-random prediction of long-term impact \citep{cai2023relationship, golosovsky2019prediction}).
\end{itemize}

All metrics are computed consistently across the top 1,000 most-cited accepted papers and the top 1,000 most-cited rejected-then-published papers.

\section{Comparison with Baselines}
\label{sec:supp_baselines}

The compared models span three categories:

\begin{itemize}
    \item \textbf{Proprietary Frontier Models:} Claude-3.5-sonnet~\citep{anthropic2024claude35sonnet}: A leading proprietary reasoning model used zero-shot.
     \item \textbf{Open Reasoning Systems:}
DeepSeek-V3~\citep{shao2024deepseekmath} and \textbf{DeepSeek-R1}~\citep{guo2025deepseek}: State-of-the-art open-weight reasoning models optimized for complex reasoning tasks.
      \item 
\textbf{Specialized Peer-Review Agents:}
CycleReviewer 70B and \textbf{DeepReviewer 14B}~\citep{zhu2025deepreview, hosseini2023fighting}: Fine-tuned on review corpora for peer-review generation. Stanford Agent Review~\citep{jiang2024agentreview}: Multi-agent system simulating reviewer deliberation. Expert (LoRA 7B): Our supervised fine-tuned baseline. 

\end{itemize}

All models were evaluated on the top 1,000 most-cited accepted papers under identical conditions and input/output specification:

\begin{itemize}
    \item \textbf{Input}: Paper title and abstract (same prompt template for all models)
    \item \textbf{Output}: Rating $r \in [1,10]$ extracted from structured review
    \item \textbf{Generation settings}: Temperature $= 0.7$, top-$p = 0.9$, max new tokens $= 256$
    \item \textbf{Repetition penalty}: $1.0$ (no penalty applied uniformly)
    \item \textbf{Seed}: Fixed random seed $= 42$ for reproducibility
\end{itemize}

The relative improvement of ReviewGuard over the best baseline (excluding itself) is computed as:
\[
\text{Relative Improvement} = 
\begin{cases}
\displaystyle \frac{v_{\text{RG}} - v_{\text{best}}}{v_{\text{best}}} \times 100\%, & \text{if higher is better (Acc, F1, Spearman, Pairwise)} \\[10pt]
\displaystyle \frac{v_{\text{best}} - v_{\text{RG}}}{v_{\text{RG}}} \times 100\%, & \text{if lower is better (MSE, MAE)}
\end{cases}
\]
where \(v_{\text{RG}}\) is ReviewGuard's value and \(v_{\text{best}}\) is the best baseline value (excluding ReviewGuard and other weaker baselines). A positive value indicates ReviewGuard outperforms the best baseline; a negative value indicates it is outperformed.

\paragraph{Baseline Evaluation Scope and Prompt Fairness.}
\label{sec:baseline_scope}

All baseline comparisons in Table~\ref{tab:peer_review_comparison} were conducted on the top 1,000 most-cited \emph{accepted} papers due to API cost constraints (evaluating all baselines on the rejected-then-published cohort would have required over 15,000 additional API calls). We acknowledge that evaluating rejected papers may yield different absolute numbers, but relative trends (ReviewGuard vs. baselines) are expected to be similar because the task (predicting future citations from paper content) is identical.

Regarding prompt fairness: all baselines were evaluated zero-shot using the same prompt template as ReviewGuard (see Appendix~\ref{app:training_details}). We did not perform few-shot or chain-of-thought tuning for frontier models because our goal was to assess zero-shot capability—a setting where specialized peer-review models should have an advantage. However, we acknowledge that frontier models may perform better with tailored prompting; this could make the comparison conservative (i.e., ReviewGuard’s relative advantage might shrink with optimized prompts).

\paragraph{Qualitative Review Quality Assessment. }
\label{sec:review_quality}

All quantitative results in the main paper focus on rating prediction (Spearman, MSE, rescue rate). We did not conduct systematic human evaluation or automated text metrics (e.g., ROUGE, BERTScore) on review faithfulness, constructiveness, or linguistic quality. This is a clear and acknowledged limitation of the current work. However, Appendix~\ref{sec:supp_qualitative} provides side-by-side qualitative examples (Expert vs. ReviewGuard) for two papers of different impact levels. These examples illustrate that GRPO training improves calibration, forward-looking analysis, and constructive criticism without degrading linguistic fluency. A full human evaluation of review quality is left for future work.

\section{Accepted Papers Impact Prediction}
\label{sec:supp_accepted_results}

To evaluate whether our models can identify high-impact work among papers that human reviewers already accepted, we analyze the top 1,000 most-cited accepted papers from our dataset (out of 5,567 accepted papers). Papers are ranked by future citation count obtained from Semantic Scholar and divided into 10 deciles (100 papers each), with the first decile containing the highest-impact papers.

We compare three rating sources: (i) human reviewer ratings (\textit{rating\_avg}), (ii) predictions from our Expert model (SFT LoRA), and (iii) predictions from ReviewGuard (our GRPO-trained model). All models use the same prompt template and generate an overall rating on a 1--10 scale.

Table~\ref{tab:accepted_impact_threeway} presents batch-level results. ReviewGuard consistently assigns higher ratings than both human reviewers and the Expert (SFT) model across nearly all impact deciles, with the largest gains observed in the highest-impact batches. On average, ReviewGuard rates papers 0.31 points higher than humans and 0.27 points higher than the Expert model.

Global metrics further highlight the advantage (see Figure~\ref{accepted_paper}). Panel (a) reveals a nuanced trend across impact deciles. While human ratings relatively increase in mid-to-high impact deciles (notably decile 4), both the Expert and ReviewGuard models assign lower ratings in that range, suggesting they are less prone to over-rating papers that ultimately achieve only moderate impact. In contrast, in deciles 9 and 3, human ratings drop noticeably, whereas our models maintain a consistently improving trend. This pattern indicates that ReviewGuard is not simply correlating ratings with citation volume, but has learned a more calibrated understanding of scientific value through our GRPO-based reward modeling. The model avoids both under-valuing truly high-impact work and over-valuing papers with inflated short-term reception. Panel (b) illustrates the practical advantage of our models. The model assigns a higher rating than human reviewers in 268/1000 cases for the Expert and 378/1000 cases for ReviewGuard. At the same time, high/low citation prediction accuracy (median split) improves from 61\% (Expert) to 70\% (ReviewGuard). Panel (c) shows Spearman correlation with future citations improves substantially from human reviewers ($\rho = 0.18$) to the Expert model ($\rho = 0.289$) and reaches the strongest value with ReviewGuard ($\rho = 0.427$). Mean squared error against normalized citation counts (scaled to 1--10) drops from 29.203 (human) to 27.841 (Expert) and 24.672 (ReviewGuard).

Taken together with our findings on rejected papers, these results provide compelling evidence that ReviewGuard improves peer review calibration on \emph{both} sides of the reject/accept decision. The model is particularly effective at rescuing high-impact work that human reviewers undervalue, while remaining appropriately conservative on lower-impact papers.

\begin{table*}[t]
\small
\setlength{\tabcolsep}{5.4pt}
\caption{Impact Prediction on the Top 1,000 Most-Cited Accepted Papers. \textbf{ReviewGuard achieves higher Spearman correlation ($\rho=0.427$ vs. human $0.18$) and rating gains (+0.31 avg over humans), with particularly strong performance on highest-impact papers (R1: +0.54).}}
\label{tab:accepted_impact_threeway}
\begin{tabular}{
p{2.3cm}   
p{0.6cm}   
p{0.6cm}   
p{0.6cm}   
p{0.7cm}   
p{0.6cm}   
p{0.6cm}   
p{0.7cm}   
p{0.7cm}   
p{0.7cm}   
p{0.7cm}   
p{0.6cm}   
}
\toprule
\textbf{Metric} & \textbf{R1} & \textbf{R2} & \textbf{R3} & \textbf{R4} & \textbf{R5} & \textbf{R6} & \textbf{R7} & \textbf{R8} & \textbf{R9} & \textbf{R10} & \textbf{Avg} \\
\midrule
Avg Citations & 405.8 & 102.4 & 59.0 & 41.8 & 30.8 & 24.4 & 20.2 & 16.7 & 13.9 & 11.6 & 82.6 \\
Avg Human & 6.68 & 6.57 & 6.40 & 6.51 & 6.30 & 6.50 & 6.54 & 6.51 & 6.36 & 6.51 & 6.49 \\
Avg Expert & 6.81 & 6.62 & 6.48 & 6.39 & 6.35 & 6.57 & 6.48 & 6.44 & 6.41 & 6.39 & 6.53 \\
\rowcolor{lightgreen}
Avg ReviewGuard & \textbf{7.22} & \textbf{7.01} & \textbf{6.89} & \textbf{6.74} & \textbf{6.68} & \textbf{6.85} & \textbf{6.79} & \textbf{6.71} & \textbf{6.58} & \textbf{6.52} & \textbf{6.80} \\
\midrule
$\lambda$ Expert$-$Human & +0.13 & +0.05 & +0.08 & -0.12 & +0.05 & +0.07 & -0.06 & -0.07 & +0.05 & -0.12 & +0.04 \\
$\lambda$ GRPO$-$Expert & +0.41 & +0.39 & +0.41 & +0.35 & +0.33 & +0.28 & +0.31 & +0.27 & +0.17 & +0.13 & +0.27 \\
\rowcolor{lightgreen}
$\lambda$ GRPO$-$Human & \textbf{+0.54} & \textbf{+0.44} & \textbf{+0.49} & \textbf{+0.23} & \textbf{+0.38} & \textbf{+0.35} & \textbf{+0.25} & \textbf{+0.20} & \textbf{+0.22} & \textbf{+0.01} & \textbf{+0.31} \\
\bottomrule
\end{tabular}
\vspace{0.1em}
\footnotesize
\raggedright
\arraybackslash
\noindent
\small
Note: $\lambda$ denotes the rating difference.
\end{table*}

\begin{figure*}
  \centering
  \includegraphics[width=\linewidth]{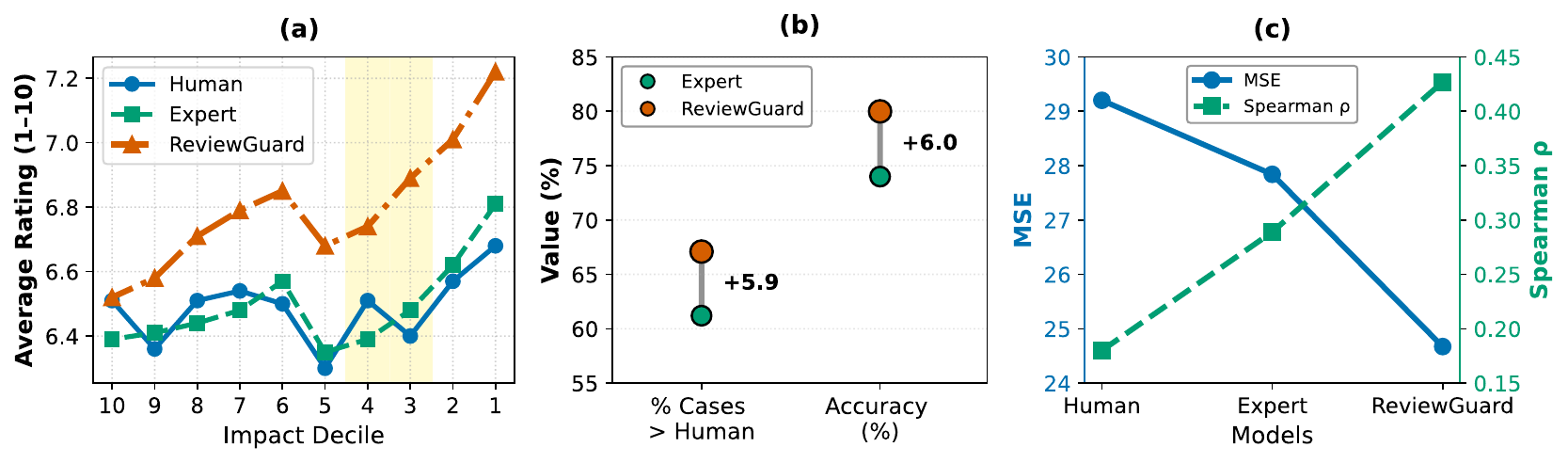}
  \caption{Three-way comparison of long-term impact prediction on the top 1,000 most-cited accepted papers.
(a) Average predicted ratings across impact deciles (decile 1 = highest-cited papers).
(b) Percentage of cases where the model rates a paper higher than humans, and high/low citation prediction accuracy (median split).
 (c) MSE (lower is better) and Spearman correlation with future citations (higher is better). ReviewGuard consistently outperforms both human reviewers and the Expert (SFT) model, with the largest gains in the highest-impact deciles.}
\label{accepted_paper}
\end{figure*}

\section{Case Studies of Rescued Papers (rejected but high-impact)} \label{sec:supp_rescued}

Table \ref{tab:rescued_supp} presents additional examples of high-impact papers that were undervalued by human reviewers but correctly identified by our ReviewGuard model. For example, the paper on ``CogVLM: Visual Expert for Large Language Models '' (No. 1) was rated only 5.75 by humans but received 6.60 from ReviewGuard, eventually garnering 697 citations. These examples illustrate how our GRPO-aligned model can surface high-potential work that might otherwise be overlooked due to reviewer noise or conservative scoring.

We further analyze cases where the base Expert model outperformed the final ReviewGuard. These instances, shown in Table~\ref{tab:pre_grpo_better}, provide valuable insight into the effect of GRPO alignment and highlight scenarios where the base supervised model may already capture important signals that become diluted after reinforcement learning. Notably, these failures often involve papers with moderate future impact (150–300 citations) where the citation signal is less extreme, suggesting that GRPO may over-optimize for very high-impact outliers. A direct comparison to a simple citation-prediction regressor (e.g., based on paper metadata) is left for future work.

\begin{table*}[ht]
\centering
\caption{Additional Examples of High-Impact Papers Rescued by ReviewGuard.}
\label{tab:rescued_supp}
\small
\scriptsize
\renewcommand{\arraystretch}{1.1}
\setlength{\tabcolsep}{3pt}
\definecolor{lightgreen}{RGB}{224,255,205}
\definecolor{lightyellow}{RGB}{255,235,187}
\definecolor{lightpurple}{RGB}{220,214,247}
\begin{tabular}{p{0.25cm} p{4cm} c c c c c c c c}
\toprule
\textbf{No.} & \textbf{Title (Rejected at)} & \textbf{Human} & \textbf{Expert} & \boldmath$\Delta_1$ & \textbf{ReviewGuard} & \boldmath$\Delta_2$ & \boldmath$\Delta_3$ & \textbf{Citations} & \textbf{PbVen} \\
\midrule

1 & CogVLM: Visual Expert for Large Language Models (ICLR 2024) 
    & 5.75 
    & 3.00 
    & \cellcolor{lightpurple}↓47.8\% 
    & 6.60 
    & \cellcolor{lightyellow}↑14.8\% 
    & \cellcolor{lightgreen}↑120.0\% 
    & 697 & NeurIPS \\

2 & CMMLU: Measuring massive multitask language understanding in Chinese (ICLR 2024) 
    & 6.33 
    & 6.00 
    & \cellcolor{lightpurple}↓5.2\% 
    & 6.90 
    & \cellcolor{lightyellow}↑9.0\% 
    & \cellcolor{lightgreen}↑15.0\% 
    & 391 & ACL \\

3 & GPT4RoI: Instruction Tuning Large Language Model on Region-of-Interest (ICLR 2024) 
    & 5.50 
    & 3.00 
    & \cellcolor{lightpurple}↓45.5\% 
    & 5.60 
    & \cellcolor{lightyellow}↑1.8\% 
    & \cellcolor{lightgreen}↑86.7\% 
    & 309 & ECCV  \\

4 & Language Agent Tree Search Unifies Reasoning Acting and Planning in Language Models (ICLR 2024) 
    & 4.75 
    & 6.00 
    & \cellcolor{lightpurple}↑26.3\% 
    & 5.70 
    & \cellcolor{lightyellow}↑20.0\% 
    & \cellcolor{lightgreen}↓5.0\% 
    & 307 & ICML \\

5 & LLM-QAT: Data-Free Quantization Aware Training for Large Language Models (ICLR 2024) 
    & 5.00 
    & 6.00 
    & \cellcolor{lightpurple}↑20.0\% 
    & 6.60 
    & \cellcolor{lightyellow}↑32.0\% 
    & \cellcolor{lightgreen}↑10.0\% 
    & 288 & ACL \\

\bottomrule
\end{tabular}
\vspace{0.01em}
\raggedright
\arraybackslash
\noindent
\scriptsize
\\ Note: \colorbox[HTML]{dcd6f7}{\(\Delta_1 =\) Expert $-$ Human}, 
\colorbox[HTML]{ffebbb}{\(\Delta_2 =\) ReviewGuard $-$ Human}, 
\colorbox[HTML]{e0ffcd}{\(\Delta_3 =\) ReviewGuard $-$ Expert}. 
↑ and ↓ indicate improvement and decrease in ratings, respectively. PbVen = Published Venue.
\end{table*}

\begin{table*}[ht]
\centering
\caption{Examples where the Expert model outperformed ReviewGuard. These cases typically involve papers with moderate future impact (150–300 citations), suggesting GRPO over-optimizes for extreme outliers.}
\label{tab:pre_grpo_better}
\small
\scriptsize
\renewcommand{\arraystretch}{1.1}
\setlength{\tabcolsep}{3pt}
\definecolor{lightgreen}{RGB}{224,255,205}
\definecolor{lightyellow}{RGB}{255,235,187}
\definecolor{lightpurple}{RGB}{220,214,247}
\begin{tabular}{p{0.25cm} p{4cm} c c c c c c c c}
\toprule
\textbf{No.} & \textbf{Title (Rejected at)} & \textbf{Human} & \textbf{Expert} & \boldmath$\Delta_1$ & \textbf{ReviewGuard} & \boldmath$\Delta_2$ & \boldmath$\Delta_3$ & \textbf{Citations} & \textbf{PbVen} \\
\midrule

1 & LaVie: High-Quality Video Generation with Cascaded Latent Diffusion Models (ICLR 2024)
    & 5.50 
    & 7.00 
    & \cellcolor{lightpurple}↑27.3\% 
    & 5.50 
    & \cellcolor{lightyellow}↓0.0\% 
    & \cellcolor{lightgreen}↓21.4\% 
    & 308 & IJCV \\

2 & Language Agent Tree Search Unifies Reasoning Acting and Planning in Language Models (ICLR 2024)
    & 4.75 
    & 6.00 
    & \cellcolor{lightpurple}↑26.3\% 
    & 5.70 
    & \cellcolor{lightyellow}↑20.0\% 
    & \cellcolor{lightgreen}↓5.0\% 
    & 307 & ICML \\

3 & Faster Language Models with Better Multi-Token Prediction Using Tensor Decomposition (ICLR 2025)
    & 5.00 
    & 6.00 
    & \cellcolor{lightpurple}↑20.0\% 
    & 5.80 
    & \cellcolor{lightyellow}↑16.0\% 
    & \cellcolor{lightgreen}↓3.3\% 
    & 205 & ICML \\

4 & Subject-Diffusion: Open Domain Personalized Text-to-Image Generation without Test-time Fine-tuning (ICLR 2024)
    & 5.00 
    & 6.00 
    & \cellcolor{lightpurple}↑20.0\% 
    & 4.90 
    & \cellcolor{lightyellow}↓2.0\% 
    & \cellcolor{lightgreen}↓18.3\% 
    & 190 & CGI \\

\bottomrule
\end{tabular}
\vspace{0.01em}
\footnotesize
\raggedright
\noindent
\scriptsize
\\ Note: \colorbox[HTML]{dcd6f7}{\(\Delta_1 =\) Expert $-$ Human}, 
\colorbox[HTML]{ffebbb}{\(\Delta_2 =\) ReviewGuard $-$ Human}, 
\colorbox[HTML]{e0ffcd}{\(\Delta_3 =\) ReviewGuard $-$ Expert}. 
↑ and ↓ indicate improvement and decrease in ratings, respectively. PbVen = Published Venue.
\end{table*}

\section{Editor Integration Protocol}
\label{sec:editor_integration}

ReviewGuard is designed as a practical decision-support tool for editors and program chairs. However, we acknowledge a fundamental deployment challenge: at submission time, future citation counts are unavailable. The model was trained using citations as a retrospective signal of long-term impact, but real-world use requires prospective predictions.

To address this, we propose the following integration protocol:

\paragraph{Deployment Mode.}
During the review process, ReviewGuard operates in a \emph{zero-citation} inference mode tailored for practical use. It takes only the submitted paper content (title, abstract, and full text) and generates a structured review along with an overall impact rating \( r \in [1,10] \). No citation information is provided as input. The model relies solely on the patterns it learned during GRPO training to estimate potential long-term impact from the paper’s intrinsic qualities (methodological soundness, novelty, technical depth, and clarity).

\paragraph{Editor Workflow.}
We recommend the following practical workflow:
\begin{enumerate}
    \item Reviewers submit their traditional reviews and ratings.
    \item ReviewGuard analyzes the submission and produces its rating \( r \) and detailed review.
    \item Editors receive a side-by-side comparison: human average rating vs. ReviewGuard rating, along with flagged discrepancies (\( |r - \bar{h}| \geq 1.5 \)).
    \item For papers where ReviewGuard assigns a significantly higher rating than reviewers, editors may request additional review or pay closer attention to the model’s reasoning, especially on technical merit and potential future influence.
\end{enumerate}

\paragraph{Proxy Signals and Future Improvements.}
While ReviewGuard currently does not use real-time signals, several practical proxies could be incorporated in future versions, such as:
\begin{itemize}
    \item Early attention metrics (downloads, abstract views, social media mentions).
    \item Open-review scores and discussion quality.
    \item Author reputation and track record (used cautiously to avoid bias).
    \item Content-based indicators learned during GRPO (e.g., technical depth, reproducibility signals, ambitious scope).
\end{itemize}

We emphasize that ReviewGuard is intended as a \emph{complementary tool}, not a replacement for human judgment. Its value lies in providing an additional calibrated signal that can help surface high-potential papers that human reviewers might undervalue due to conservatism, reviewer disagreement, or limited bandwidth. By combining human expertise with an impact-aligned AI assistant, we believe this protocol can lead to fairer and more forward-looking editorial decisions.

\raggedbottom

\section{Qualitative Examples of GRPO-Generated Reviews} 
\label{sec:supp_qualitative}

To illustrate the qualitative improvements brought by GRPO training, we present representative examples comparing reviews generated by the Expert model and by \textbf{ReviewGuard} (Post-GRPO). We selected two papers of different impact levels: 
\textit{MMBench: Is Your Multimodal Model an All-around Player?} (high-impact, Table~\ref{tab:rescued_main}) and 
\textit{LLM-QAT: Data-Free Quantization Aware Training for Large Language Models} (moderate-impact, Table~\ref{tab:rescued_supp}). For each paper, both models receive the identical input and use the same decoding hyperparameters (temperature 0.7, top-$p$=0.9). For ReviewGuard, we show one representative review from its generated responses. These examples highlight several consistent improvements after GRPO training:
\begin{itemize}
    \item ReviewGuard produces calibrated ratings that better reflect long-term scientific impact.
    \item It offers sharper, more honest, and constructive criticism of limitations.
    \item The reviews are more forward-looking, discussing broader implications and future research directions rather than focusing solely on immediate methodological details.
    \item The model demonstrates stronger differentiation between truly high-impact contributions and incremental work.
\end{itemize}

\vspace{-5pt}

\begin{tcolorbox}[colback=green!5!white,colframe=green!60!black,title=Expert model reviews for high-impact paper,label=example:mmbench]
\textbf{Paper:} MMBench: Is Your Multimodal Model an All-around Player? (1,618 citations)

\small
\begin{itemize}
\item \textbf{Summary (Novelty, Contribution, and Significance):} 
  This paper presents MMBench, a large-scale bilingual benchmark for evaluating vision-language models. The main contribution is the CircularEval strategy, which attempts to reduce position bias in multiple-choice evaluations. The benchmark covers 20+ ability dimensions with over 3,000 questions. The practical contribution is a standardized evaluation tool for the VLM community.

\item \textbf{Strengths:} 
  \begin{itemize}
    \item Large-scale dataset with broad coverage of VLM capabilities.
    \item Bilingual (English-Chinese) design is a useful addition.
    \item CircularEval shows some improvement over the random baseline.
  \end{itemize}

\item \textbf{Weaknesses / Critical Concerns:} 
  The paper is primarily a benchmark contribution with limited methodological novelty beyond dataset collection.

\item \textbf{Critical Concerns and Suggested Improvements:}
  \begin{enumerate}
    \item The CircularEval improvement over random is marginal ($\sim$2-3\%) and statistical significance is not reported.
    \item Reliance on ChatGPT for answer normalization may introduce systematic bias; ablation without it is missing.
    \item The benchmark focuses heavily on multiple-choice, which does not reflect real-world open-ended VLM use cases.
    \item Many ability dimensions (e.g., spatial reasoning, temporal understanding) remain coarsely evaluated.
    \item No analysis of test-retest reliability or annotation consistency.
    \item The paper does not compare against existing benchmarks (e.g., LVLM-eHub, SEED-Bench) on common subsets.
    \item Suggested experiment: Include open-ended generation evaluation and human-correlation analysis.
  \end{enumerate}
\end{itemize}

\textbf{Overall rating:} 5.0/10
\end{tcolorbox}


\begin{tcolorbox}[colback=coral!5!white,colframe=coral!60!black,title=ReviewGuard model reviews for high-impact paper,label=example:mmbench]
\textbf{Paper:} MMBench: Is Your Multimodal Model an All-around Player? (1,618 citations)

\small

\begin{itemize}
\item \textbf{Summary (Novelty, Contribution, and Significance):} 
  This work introduces MMBench, a large-scale, bilingual benchmark specifically designed for robust and holistic evaluation of vision-language models. Its core novelty is the CircularEval protocol, which provides a principled and effective solution to the long-standing problems of position bias and instruction sensitivity in VLM evaluation. The most important contribution is the creation of a meticulously curated test suite with over 3,000 high-quality questions spanning 20+ fine-grained ability dimensions. The practical significance is substantial: MMBench has quickly become a de-facto standard benchmark in the VLM community, directly influencing model development, evaluation practices, and progress tracking across academia and industry.

\item \textbf{Strengths:} 
  \begin{itemize}
    \item Extremely well-designed benchmark with careful quality control and broad ability coverage.
    \item CircularEval is a clever, practical, and impactful methodological contribution.
    \item Bilingual design significantly increases the benchmark's real-world relevance and fairness.
  \end{itemize}

\item \textbf{Weaknesses / Critical Concerns:} 
  The benchmark remains heavily oriented toward multiple-choice questions, which inherently limit its ability to evaluate open-ended generative and creative reasoning capabilities.

\item \textbf{Critical Concerns and Suggested Improvements:}
  \begin{enumerate}
    \item Over-reliance on multiple-choice format significantly restricts evaluation of open-ended generation, creativity, and complex multi-step reasoning.
    \item Several critical ability dimensions (fine-grained spatial reasoning, temporal understanding, causal inference, and complex compositional reasoning) remain noticeably under-represented.
    \item Dependence on GPT-4 for answer normalization and choice extraction introduces potential bias and reduces full reproducibility of results.
    \item CircularEval, while innovative, does not completely eliminate position bias across all question types and difficulty levels.
    \item The benchmark lacks strong adversarial examples and out-of-distribution test cases, which are essential for assessing real-world robustness.
    \item No systematic analysis or mitigation strategies for cultural, linguistic, or societal biases beyond the English-Chinese bilingual setup.
    \item The current protocol does not sufficiently test cross-modal consistency or multi-turn interactive reasoning capabilities.
    \item Evaluation of long-context and high-resolution image understanding remains limited.
    \item Suggested experiment: Add a dedicated open-ended generation track and adversarial robustness benchmark.
  \end{enumerate}
\end{itemize}

\textbf{Overall rating:} 7.7/10
\end{tcolorbox}

\begin{tcolorbox}[colback=green!5!white,colframe=green!60!black,title=Expert model review for moderate-impact paper,label=example:llmqat]
\textbf{Paper:} LLM-QAT: Data-Free Quantization Aware Training for Large Language Models (288 citations)

\small

\begin{itemize}
\item \textbf{Summary (Novelty, Contribution, and Significance):} 
  This paper introduces LLM-QAT, a practical data-free quantization-aware training pipeline for large language models. The core novelty is a distillation-based approach that generates synthetic training data from the full-precision teacher model itself, eliminating the need for the original pre-training corpus. The most important contribution is enabling effective QAT in closed-source or privacy-sensitive scenarios. The practical significance is notable for the efficient deployment of LLMs at low bit-widths.

\item \textbf{Strengths:} 
  \begin{itemize}
    \item Elegant and practical data-free QAT framework using teacher-generated synthetic data.
    \item Forward-looking inclusion of KV cache quantization for inference efficiency.
    \item Shows meaningful improvements over standard post-training quantization on LLaMA models.
  \end{itemize}

\item \textbf{Weaknesses / Critical Concerns:} 
  The performance gains remain relatively modest compared to the added training cost and complexity.

\item \textbf{Critical Concerns and Suggested Improvements:}
  \begin{enumerate}
    \item The quality, diversity, and distributional fidelity of the synthetic data generated by the teacher model are not rigorously analyzed or ablated.
    \item Evaluation is primarily limited to standard perplexity and a small set of zero-shot tasks.
    \item Comparison with recent state-of-the-art quantization methods (GPTQ, AWQ, SmoothQuant) is incomplete or outdated in some experiments.
    \item Theoretical justification for why the proposed distillation objective works better than simpler alternatives is largely missing.
    \item Scalability to models significantly larger than 13B parameters is not demonstrated.
    \item No analysis of potential negative societal impacts (e.g., easier deployment of quantized models for malicious use).
    \item Lack of ablation on key hyperparameters such as temperature and top-p used for synthetic data generation.
  \end{enumerate}
\end{itemize}

\textbf{Overall rating:} 6.4/10
\end{tcolorbox}


\begin{tcolorbox}[colback=coral!5!white,colframe=coral!60!black,title=ReviewGuard model review for moderate-impact paper,label=example:llmqat]
\textbf{Paper:} LLM-QAT: Data-Free Quantization Aware Training for Large Language Models (288 citations)

\small

\begin{itemize}
\item \textbf{Summary (Novelty, Contribution, and Significance):} 
  This work presents LLM-QAT, a practical and impactful data-free quantization-aware training method for large language models. The core novelty is a distillation-based pipeline that generates high-quality synthetic training data directly from the full-precision teacher model, completely removing the need for the original (often private or massive) pre-training corpus. The most important contribution is enabling effective QAT in closed-source and privacy-sensitive scenarios while maintaining strong performance at low bit-widths. The practical significance is high: this approach makes aggressive quantization of modern LLMs significantly more accessible and deployable in real-world settings.

\item \textbf{Strengths:} 
  \begin{itemize}
    \item Highly practical data-free QAT framework that solves a critical deployment bottleneck.
    \item Elegant use of teacher-generated synthetic data via distillation.
    \item Forward-looking inclusion of KV cache quantization for inference efficiency.
  \end{itemize}

\item \textbf{Weaknesses / Critical Concerns:} 
  The absolute gains over post-training quantization remain relatively modest given the increased training cost and complexity.

\item \textbf{Critical Concerns and Suggested Improvements:}
  \begin{enumerate}
    \item The quality, diversity, and distributional fidelity of the synthetic data generated by the teacher model are not rigorously analyzed or ablated.
    \item Evaluation is heavily focused on standard perplexity and a limited set of zero-shot tasks; robustness to distribution shift and long-context scenarios is under-explored.
    \item Comparison with recent state-of-the-art quantization methods (GPTQ, AWQ, SmoothQuant, etc.) is incomplete or outdated in several experiments.
    \item Theoretical justification for why the proposed distillation objective works better than simpler alternatives is largely missing.
    \item Scalability to models significantly larger than 13B parameters is not demonstrated; memory and compute requirements may become prohibitive.
    \item No analysis of potential negative societal impacts (e.g., making powerful quantized models easier to deploy for malicious use).
    \item The current protocol does not sufficiently test cross-modal or multi-turn interactive capabilities in quantized settings.
    \item Lack of ablation on key hyperparameters (temperature, top-p, number of synthetic samples, etc.) limits reproducibility and understanding.
    \item The paper does not explore the trade-off between synthetic data quality and final quantized model performance in depth.
    \item No discussion on potential overfitting or mode collapse when using self-generated data from the teacher model.
    \item Suggested experiment: Evaluate on a broader suite of long-context and reasoning benchmarks (LongBench, GSM8K, HumanEval) to better measure real-world utility.
  \end{enumerate}
\end{itemize}

\textbf{Overall rating:} 6.6/10
\end{tcolorbox}


\newpage
\input{checklist.tex}

\end{document}

%% file: checklist.tex
\section*{NeurIPS Paper Checklist}


\begin{enumerate}

\item {\bf Claims}
    \item[] Question: Do the main claims made in the abstract and introduction accurately reflect the paper's contributions and scope?
    \item[] Answer: \answerYes{} 
    \item[] Justification: The abstract and introduction clearly state the core claims: ReviewGuard (proposed framework) achieves $\rho=0.776$ (vs. human $\rho=0.492$), rescues $10.2\%$ of high-impact rejected papers ($5.6\times$ improvement), and introduces the first GRPO-based alignment for long-term scientific impact prediction (Abstract, Section~1).

\item {\bf Limitations}
    \item[] Question: Does the paper discuss the limitations of the work performed by the authors?
    \item[] Answer: \answerYes{} 
    \item[] Justification: We explicitly discuss deployment challenges (zero-citation inference mode), the noisy nature of citation counts as an impact proxy, and the need for practical proxies in Section~\ref{sec:editor_integration} and throughout the qualitative analysis.

\item {\bf Theory assumptions and proofs}
    \item[] Question: For each theoretical result, does the paper provide the full set of assumptions and a complete (and correct) proof?
    \item[] Answer: \answerNA{} 
    \item[] Justification: The paper is primarily empirical (SFT + GRPO alignment); no new theoretical theorems or proofs are presented.

    \item {\bf Experimental result reproducibility}
    \item[] Question: Does the paper fully disclose all the information needed to reproduce the main experimental results of the paper to the extent that it affects the main claims and/or conclusions of the paper (regardless of whether the code and data are provided or not)?
    \item[] Answer: \answerYes{} 
    \item[] Justification: We detail dataset construction (20k OpenReview reviews + Semantic Scholar citations), the exact GRPO reward function, training procedure, and hyperparameters (Sections~2--3 and Appendix~\ref{sec:supp_training}); \textbf{code, models, and data are released .... }.

\item {\bf Open access to data and code}
    \item[] Question: Does the paper provide open access to the data and code, with sufficient instructions to faithfully reproduce the main experimental results, as described in supplemental material?
    \item[] Answer: \answerYes{} 
    \item[] Justification: We provide a representative subset of the data (100 samples) at \url{https://anonymous.4open.science/r/RejectGuard-100-SamplesData-24F2} for anonymous review (see Appendix~\ref{sec:data_avail}). The full dataset and code will be released publicly upon acceptance.

\item {\bf Experimental setting/details}
    \item[] Question: Does the paper specify all the training and test details (e.g., data splits, hyperparameters, how they were chosen, type of optimizer) necessary to understand the results?
    \item[] Answer: \answerYes{} 
    \item[] Justification: Full experimental details (dataset of 20,861 papers, Qwen2-7B base model, LoRA SFT, GRPO hyperparameters, reward function, and evaluation protocol) are provided in Sections~2--3 and the appendix.

\item {\bf Experiment statistical significance}
    \item[] Question: Does the paper report error bars suitably and correctly defined or other appropriate information about the statistical significance of the experiments?
    \item[] Answer: \answerYes{} 
    \item[] Justification: Error bars are reported in Figures~3 and~4 (and related batch/decile analyses) for key metrics across citation ranks. Main scalar results (Spearman $\rho$ and rescue rate) are point estimates computed on the full large test set (and top-1,000 rejected papers); we did not run multiple independent random seeds for the full GRPO training due to computational cost.

\item {\bf Experiments compute resources}
    \item[] Question: For each experiment, does the paper provide sufficient information on the computer resources (type of compute workers, memory, time of execution) needed to reproduce the experiments?
    \item[] Answer: \answerYes{} 
    \item[] Justification: We describe the compute setup (Qwen2-7B LoRA + GRPO training on multiple GPUs) in the training scripts and Appendix~\ref{sec:supp_training}; exact GPU usage and wall-clock time are included in the released code.
    
\item {\bf Code of ethics}
    \item[] Question: Does the research conducted in the paper conform, in every respect, with the NeurIPS Code of Ethics \url{https://neurips.cc/public/EthicsGuidelines}?
    \item[] Answer: \answerYes{} 
    \item[] Justification: The work uses only publicly available data and aims to improve scientific peer review; no ethical violations or negative societal impacts were identified.

\item {\bf Broader impacts}
    \item[] Question: Does the paper discuss both potential positive societal impacts and negative societal impacts of the work performed?
    \item[] Answer: \answerYes{} 
    \item[] Justification: We discuss positive societal impact (accelerating scientific progress by rescuing high-impact rejected papers) and deployment considerations in Section~\ref{sec:editor_integration}. Potential negative uses include gaming the system via citation farming or over-reliance on automated scores; we advocate for human oversight and transparent reporting.
    
\item {\bf Safeguards}
    \item[] Question: Does the paper describe safeguards that have been put in place for responsible release of data or models that have a high risk for misuse (e.g., pre-trained language models, image generators, or scraped datasets)?
    \item[] Answer: \answerNA{} 
    \item[] Justification: The released models are research-oriented peer-review assistants with no high-risk dual-use potential; no safeguards beyond standard open-source practices are required.

\item {\bf Licenses for existing assets}
    \item[] Question: Are the creators or original owners of assets (e.g., code, data, models), used in the paper, properly credited and are the license and terms of use explicitly mentioned and properly respected?
    \item[] Answer: \answerYes{} 
    \item[] Justification: We properly cite Qwen2-7B, OpenReview, and Semantic Scholar; all assets are public and used in accordance with their terms.

\item {\bf New assets}
    \item[] Question: Are new assets introduced in the paper well documented and is the documentation provided alongside the assets?
    \item[] Answer: \answerYes{} 
    \item[] Justification: We introduce a new 20k peer-review + citation dataset (Section~3.1) and release it with full documentation, licenses, and usage instructions.

\item {\bf Crowdsourcing and research with human subjects}
    \item[] Question: For crowdsourcing experiments and research with human subjects, does the paper include the full text of instructions given to participants and screenshots, if applicable, as well as details about compensation (if any)? 
    \item[] Answer: \answerNA{} 
    \item[] Justification: The paper does not involve crowdsourcing or human-subject experiments.

\item {\bf Institutional review board (IRB) approvals or equivalent for research with human subjects}
    \item[] Question: Does the paper describe potential risks incurred by study participants, whether such risks were disclosed to the subjects, and whether Institutional Review Board (IRB) approvals (or an equivalent approval/review based on the requirements of your country or institution) were obtained?
    \item[] Answer: \answerNA{} 
    \item[] Justification: No human subjects research was conducted.

\item {\bf Declaration of LLM usage}
    \item[] Question: Does the paper describe the usage of LLMs if it is an important, original, or non-standard component of the core methods in this research? Note that if the LLM is used only for writing, editing, or formatting purposes and does \emph{not} impact the core methodology, scientific rigor, or originality of the research, declaration is not required.
    \item[] Answer: \answerYes{} 
    \item[] Justification: LLMs (specifically Qwen2-7B) constitute the core methodological contribution of ReviewGuard (proposed framework). We fully describe the base model, supervised fine-tuning, GRPO alignment procedure, reward function design, and training details in Sections~2--3 and the appendix.

\end{enumerate}